\documentclass[a4paper, 10pt,reqno]{amsart}
\usepackage{amssymb}
\usepackage{amsmath}
\usepackage{amsthm}
\usepackage{latexsym}
\usepackage{mathrsfs}
\usepackage{bm}
\usepackage{eucal}
\usepackage{tipa}
\usepackage{slashed}
\usepackage{mathtools}
\usepackage{enumitem}

\usepackage{tikz-cd}

\usepackage{color}

\usepackage[boxsize=0.5em,aligntableaux=top]{ytableau}
\newcommand{\yd}{\ydiagram}

\usepackage[numbers]{natbib}

\usepackage{hyperref}

\usepackage{orcidlink}

\hypersetup{
    colorlinks=true,
    linkcolor=blue,
    filecolor=blue,      
    urlcolor=blue,
}

\newcounter{mnotecount}[section]

\allowdisplaybreaks[2]

\theoremstyle{plain}
\newtheorem{theorem}{Theorem}
\newtheorem{proposition}[theorem]{Proposition}
\newtheorem{lemma}[theorem]{Lemma}

\newtheorem{definition}[theorem]{Definition}
\newtheorem{remark}[theorem]{Remark}

\def\sDiv{\mathscr{D}}
\def\sCurl{\mathscr{C}}
\def\sCurlDagger{\mathscr{C}^\dagger}
\def\sTwist{\mathscr{T}}

\def\ImA{\Im\hat{\mathcal{A}}}

\def\SymSpin{\mathcal{S}}
\def\Vec{V}
\def\linWeyl{\mathcal{W}}
\def\Hop{\mathsf{H}}
\def\Cop{\mathsf{C}}
\def\Dop{\mathsf{D}}

\def\Kop{\mathsf{K}}
\def\Lop{\mathsf{L}}
\def\Pop{\mathsf{P}}
\def\Qop{\mathsf{Q}}
\def\Fop{\mathsf{F}}
\def\id{\mathrm{id}}
\def\defeq{\coloneqq}

\def\FrameCon{\gamma} 
\def\FlatCon{\Upsilon} 

\def\InvSymb{\mathbb{I}}

\newcommand{\Mcal}{\mathcal{M}}

\numberwithin{equation}{section}

\tikzcdset{row sep/normal=2cm}
\tikzcdset{column sep/normal=2cm}
\tikzset{ampersand replacement=\&}

  \newcommand{\del}{\partial}

  \newcommand{\Lie}{\mathcal{L}}

\renewcommand{\d}{\mathrm{d}}

  \newcommand{\Secs}{\Gamma}
  \newcommand{\DD}{\mathbb{D}}

\title{Compatibility complex for black hole spacetimes}

\author[S. Aksteiner]{Steffen Aksteiner \orcidlink{0000-0002-0009-4292}}
\email{steffen.aksteiner@aei.mpg.de}
\address{Albert Einstein Institute, Am M\"uhlenberg 1, D-14476 Potsdam, Germany }

\author[L. Andersson]{Lars Andersson \orcidlink{0000-0002-6364-7384}}

\email{lars.andersson@aei.mpg.de}
\address{Albert Einstein Institute, Am M\"uhlenberg 1, D-14476 Potsdam, Germany }

\author[T. B\"{a}ckdahl]{Thomas B\"{a}ckdahl \orcidlink{0000-0003-3240-2445}}

\email{thomas.backdahl@chalmers.se}

\address{Mathematical Sciences, Chalmers University of Technology and University of Gothenburg, SE-412~96 Gothenburg, Sweden}

\address{Albert Einstein Institute, Am M\"uhlenberg 1, D-14476 Potsdam, Germany }

\author[I. Khavkine]{Igor Khavkine \orcidlink{0000-0003-4255-6579}}

\email{khavkine@math.cas.cz}
\address{Institute of Mathematics of the Czech Academy of Sciences, \v{Z}itn{\'a} 25, 115 67 Praha 1, Czech Republic}

\author[B. Whiting]{Bernard Whiting \orcidlink{0000-0002-8501-8669}}

\email{bernard@phys.ufl.edu}
\address{Department of Physics, University of Florida, 2001 Museum Road, Gainesville, FL 32611-8440, USA}

\begin{document}


\begin{abstract}
The set of local gauge invariant quantities for linearized gravity on the Kerr spacetime presented by two of the authors (S.A, T.B.) in \cite{ab-kerr} is shown to be complete. In particular, any gauge invariant quantity for linearized gravity on Kerr that is local and of finite order in derivatives can be expressed in terms of these gauge invariants and derivatives thereof. The proof is carried out by constructing a complete compatibility complex for the Killing operator, and demonstrating the equivalence of the gauge invariants from \cite{ab-kerr} with the first compatibility operator from that complex.
\end{abstract}

\maketitle

\section{Introduction}
It is a fundamental principle of general relativity that physically measurable quantities are gauge invariant, in the sense that physical phenomena should not depend on the coordinates used to describe them. In the modelling of gravitational radiation emission from the binary inspiral and merger of two compact objects, such as black holes and neutron stars, one of the most important outcomes is the waveform extracted near infinity, which is what can be detected in gravitational wave observatories. Thus it is imperative to represent such a waveform in terms of gauge invariant quantities.
Even in fully numerical approaches to waveform computation, a waveform can typically be described as a perturbation away from some asymptotically flat (reference) background spacetime.  Thus, gauge invariant asymptotic waveforms can actually be obtained by analyzing perturbations. 

In this paper, we investigate local gauge invariants for first order perturbations of the Kerr spacetime background, and describe a set which is complete in a sense that we make clear. Ours is not the first attempt to describe perturbative gauge invariants on black hole spacetimes. See \cite{2017CQGra..34q4001T,swaab,ab-kerr,kh-compat,fhk,1999GReGr..31.1855J,martel-poisson,mobpm-kerr} and references therein for earlier work. 
See also \cite{stewart-walker} for a discussion of coordinate and tetrad gauge dependence. 
Here we shall 
 rely on the methods introduced in \cite{kh-compat}, applied to the Kerr geometry.
Proofs of completeness of a set of gauge invariants are a relatively recent development. They have been given for a small number of other spacetime reference backgrounds including the spherically symmetric Schwarzschild spacetime and the conformally flat Friedmann-Robertson-Walker spacetimes \cite{fhh, fhk, kh-compat}.  
Nevertheless, this paper is the first to fully demonstrate completeness for a set of gauge invariants for the Kerr spacetime. See  \cite{pommaret1,pommaret2} for work on related problems.

In order to solve the problem of classifying all local gauge invariants for linearized gravity on the Kerr spacetime, it has been necessary to apply techniques and results that are not in common use in general relativity. Although the construction of the gauge invariants uses methods that have previously been applied in the literature on black hole perturbations, cf.~\cite{ab-kerr} and Remark~\ref{rmk:intuition} below for further explanation, a proof of their completeness requires the application of techniques and results from the theory of differential complexes.

The analysis of gauge invariant quantities is particularly important from the point of view of applications in gravitational wave analysis, partly because most compact binary mergers result in a Kerr black hole, and partly because, in the case of a binary with an extreme mass ratio, it is not yet known how to express the waveform representing gravitational wave emission. Current efforts to tackle this problem require evaluating the (covariant but gauge dependent \cite{Barack:2001ph,Pound:2013faa}) 
gravitational self-force to second order in the mass ratio, and it is anticipated that the gauge invariants introduced in \cite{ab-kerr} and shown here to be complete will prove useful in that evaluation (just as the mode-decoupled gauge invariants of \cite{2017CQGra..34q4001T} have proved useful at first order \cite{Thompson:2018lgb}). We mention that the previously considered set of gauge invariants for linearized gravity on the Kerr spacetime, presented in \cite{ab-kerr}, includes the set of gauge invariants in \cite{mobpm-kerr} as a strict subset. 

\subsection*{Motivation and background.}
Several problems have served as major motivations for the development of black hole perturbation theory during the last half-century. Among these are the self-force problem mentioned above and the closely related black hole stability problem. The Teukolsky scalars, which are two of the gauge invariants for linearized gravity on the Kerr background under consideration here, play a central role in the recent proof of linear stability of the Kerr black hole \cite{2019arXiv190303859A}.  

Let $(\Mcal, g_{ab})$ be a member of the Kerr family of stationary, rotating vacuum spacetimes and let $\Kop$ denote the Killing operator, 
\begin{align} 
(\Kop v)_{ab} = \Lie_v g_{ab} = \nabla_a v_b + \nabla_b
v_a.
\end{align}
Due to the covariant nature of the Einstein equations, given a solution $h_{ab}$ of  the linearized Einstein equations on $(\Mcal, g_{ab})$, the perturbation 
\begin{equation} \label{eq:hhprimtransf}
h'_{ab} = h_{ab} + (\Kop v)_{ab}
\end{equation} 
is also a solution. Any two metric perturbations are equivalent up to gauge when they differ by the
image of the Killing operator, and in particular represent  physically equivalent states. The linearized metric $h_{ab}$ is highly gauge dependent. Therefore, in order to extract the physical information of $h_{ab}$, it is necessary to either introduce gauge conditions (which introduces further ambiguity),
or to find gauge invariant quantities. 

By a gauge invariant quantity, we here mean a compatibility operator for the Killing operator, i.e. a  covariant linear differential operator $\Qop$ on symmetric 2-tensors \eqref{eq:hhprimtransf} satisfying
\begin{align} 
\Qop h' = \Qop h.
\end{align}
Equivalently in operator form it can be formulated as
\begin{align} 
\Qop \circ \Kop = 0 .
\end{align}
In general, we consider compatibility operators taking values in suitable vector bundles, and therefore we may view the components of $\Qop$ as a list of scalar compatibility operators $\Qop = (\Qop_1, \dots, \Qop_N)$. 

A compatibility operator $\Qop$ is complete if it contains the complete information of $h_{ab}$ modulo gauge, i.e. if 
\begin{align} \label{eq:comp1a}
\Qop h = 0
\end{align} 
only if $h_{ab}$ is locally pure gauge, i.e. if $h_{ab}$ is locally in the image of the Killing operator, 
\begin{align} \label{eq:comp1b}
h_{ab} = (\Kop v)_{ab}
\end{align} 
for some vector field $v$. 

It is a remarkable fact that a complete set of gauge invariants on the
Kerr background not only exists, but is finite and can be
given by explicit formulas of manageable complexity. In fact, a candidate
complete list was recently announced for the first time in \cite{ab-kerr}. In
this work, we prove its completeness using the
methods of \cite{kh-compat}. More precisely, we construct a compatibility operator $\Kop_1$ for $\Kop$ such that any other compatibility operator
$\Lop$ factors through it, i.e. there exists an ${\Lop}'$ such that
\begin{align} \label{eq:comp2}
\Lop = {\Lop}'\circ \Kop_1 .
\end{align}
The two notions of completeness, from~\eqref{eq:comp1a}, \eqref{eq:comp1b} and from~\eqref{eq:comp2}, are equivalent for $\Kop$ the Killing operator on Kerr since it falls into the \emph{regular finite type} class, explained in \cite{kh-compat}. 
In the language of the theory of
over-determined partial differential equations, $\Kop_1$ is a
complete compatibility operator for the Killing operator $\Kop_0 = \Kop$, cf. Definition \ref{def:compat} below. 
In addition, we show that the components
of $\Kop_1$ themselves factor through the gauge invariants (i.e. the compatibility operator) $\tilde \Kop_1$ introduced in \cite{ab-kerr},
thus proving that $\tilde \Kop_1$ is also complete. 
In the course of the proof, we construct a sequence  
of differential operators $\Kop_l$, which successively compose to
$\Kop_{l}\circ \Kop_{l-1} = 0$, hence a \emph{complex} of differential
operators. In fact, each $\Kop_l$ will be a complete compatibility
operator for $\Kop_{l-1}$, $l>0$, that is, a \emph{compatibility}
complex for the Killing operator $\Kop$, and moreover a \emph{full} one,
meaning that it cannot be extended (though the operators degenerate to
zero after finitely many steps).

An analogous construction arises when, instead of considering the linearized 
Einstein equations, we consider Maxwell's equations, where it is well known 
that pure gauge modes are given by the exterior derivative
$A_a = (\d \phi)_a$. So the
analog of $\Kop_0$ is the exterior derivative $\d$ on $0$-forms. Then
the role of $\Kop_1$ is played by the field strength tensor $F_{ab} =
(\d A)_{ab}$, namely the exterior derivative $\d$ on $1$-forms. This
sequence extends to the de~Rham complex of exterior derivatives $\d$ on
forms of higher degrees as the full compatibility complex of the
exterior derivative $\d$ on $0$-forms. This compatibility complex is the
same on both flat and curved backgrounds, since the exterior derivative
$\d\phi$ depends only on the differential structure. This is no longer
the case for the Killing operator $\Kop$, which strongly depends on the
background metric and the compatibility complex has to be computed anew
for each background, making its construction much more challenging.

As an example of the usefulness of the higher compatibility operators
$\Kop_{l>1}$, consider again the de~Rham complex.
Just as the exterior derivatives $\d$ on higher degree forms appear in the
definition of the Hodge wave operators $\square = \delta \d + \d \delta$ on differential forms, which satisfy the
convenient identities $\square \d = \d \square$, we expect the $\Kop_1$
and $\Kop_2$ operators to appear in a similar formula for a suitable
wave-like equation satisfied by the components of $\Kop_1[h]$ when $h$
solves the linearized Einstein equations. Extending
such wave-like operators to higher nodes of the compatibility complex of
$\Kop$ gives this complex a structure reminiscent of
Hodge theory in Riemannian geometry. This Hodge-like structure will be
considered in future works, where it could have applications to the
reconstruction problem, i.e. the problem of constructing the solution to the inhomogeneous equation $\Kop_1[h] = f$, and to the computation of the cohomologies
$H^*(\Kop_l)$ with causally restricted supports or regularity
properties. In the context of constant curvature backgrounds (e.g.,\
de~Sitter spacetime), these applications have been illustrated in
\cite{kh-calabi,kh-causcohom}.

\subsection*{Remarks on methodology.}
While we have already explained the motivation for and the importance of
our results, it remains to justify our methods, which unfortunately
carry two types of technical complications. The first is the
introduction of notions from homological algebra and the formal study of
overdetermined PDEs, which are not commonly known in the mathematical
relativity literature. The second is the complexity of the formulas
needed to present our main result. The justification is simple: despite
its technical complexity, our method is the simplest one known to us to
demonstrate our main results on the completeness of gauge invariants.

There are alternative approaches to prove such completeness. One
approach is in fact an algorithmic way to construct a complete
compatibility operator $\Kop_1$ and is well-known in the literature on
overdetermined PDEs~\cite{spencer}. It has even been implemented in
computer algebra~\cite{janet}. However, its use requires all
differential operators to be explicitly expressed in coordinates and in
components, which unfortunately highly obfuscates any geometric
structure in $\Kop_1$ and, more often than not, results in extremely
long expressions of doubtful utility. A semi-algorithmic version of this
approach has been pursued in the recent papers~\cite{pommaret1,
pommaret2, pommaret3} and has yet to arrive at a full expression for a
compatibility operator $\Kop_1$, let alone one as compactly expressed as
in~\cite{ab-kerr}. Another drawback of this approach is that the
completeness of $\Kop_1$ is proved by virtue of its algorithmic
construction. Anyone interested in verifying the completeness for
themselves is forced to rerun the algorithm, which is not always
practical. The advantage of the approach in~\cite{kh-compat}, which is
related but alternative, and on which the proofs in this paper are
based, is that it allows the freedom to avoid explicit component
computations, while reducing the proof of completeness to the existence
of a clearly structured set of identities, whose structure is motivated by
homological algebra. While the presentation of these identities in
section~\ref{sec:kerr} may be daunting, its complexity is necessary, as
the spinor calculus of~\cite{ABB:symop:2014CQGra..31m5015A,
2016arXiv160106084A} actually provides the most compact way known to us
of expressing them.

Another potential approach to the construction of a complete
compatibility operator $\Kop_1$ relies on representation theory.
So-called \emph{BGG complexes}~\cite{css-bgg, cd-bgg} may be constructed
and proven to be complete compatibility complexes in a purely
representation-theoretic way on spacetimes with a transitive isometry
group. An example where this method is successful for the Killing
operator $\Kop_0$ is the de~Sitter background~\cite{css-bgg, cd-bgg,
kh-calabi}. The isometry orbits on $4$-dimensional Kerr are only
$2$-dimensional, meaning that the symmetry is definitely not transitive.
Unfortunately, in such a case, the BGG construction gives a sequence of
operators $\Kop_i$ that fail to compose to zero, that is, $\Kop_{i+1}
\circ \Kop_i \ne 0$ in general, i.e., they fail to form a complex. In addition, there
is no known systematic way of correcting this sequence to a true
compatibility complex that is different from the algorithmic approach
described in the previous paragraph. The BGG construction is an
interesting starting point, but it is currently an open question whether
it can be used to construct even a complete $\Kop_1$ while respecting
the geometric structure of the Kerr background.

\subsection*{Overview of this paper}  Section \ref{sec:compat} introduces the basic notions in the theory of compatibility complexes, and states some basic facts which shall be needed. In section \ref{sec:prel} we introduce some notations and definitions which shall be used for the proofs, including spinor calculus and characterizations of the Kerr spacetime. Section \ref{sec:kerr} contains the statement and proof of our main result, and section \ref{sec:equiv} contains a discussion of the relationship between the gauge invariants constructed in section \ref{sec:kerr} and those introduced in \cite{ab-kerr}. The longer equations of these relations are given in appendix~\ref{sec:AppC2tildeComps}. A discussion of the differences in the number of invariants and their differential order, between different sets of invariants, is given in section \ref{sec:counting}. Finally, a brief discussion of the significance of these results, applications, and future directions, is given in section \ref{sec:discussion}.

\section{Compatibility operators} \label{sec:compat}
We briefly recall here some definitions and results
from~\cite{kh-compat}, which will be referred to in
Section~\ref{sec:kerr}, where our main results will be presented.

Whenever speaking of differential operators, we will specifically mean a
linear differential operator with smooth coefficients acting on smooth
functions. More precisely, we will consider differential operators that
map between sections of vector bundles, say $V_1 \to M$ and $V_2 \to M$,
on some fixed manifold $M$, $ \Kop \colon \Secs(V_1) \to \Secs(V_2)$. The
source and target bundle of a differential operator, $V_1\to M$ and
$V_2\to M$ respectively in the last example, will be considered as part
of its definition and will most often be omitted from the notation. We
will denote the composition of two differential operators $\Lop$ and $\Kop$ by
$\Lop\circ \Kop$, or simply by $\Lop \Kop$, if no confusion is possible. A
\emph{local section} of a vector bundle $V\to M$ is a section of the
restriction bundle $V|_U \to U$ for some open $U \subset M$. A local
section $v$ that solves the differential equation $\Kop v = 0$ on its
domain of definition is a \emph{local solution}.

\begin{definition} \label{def:compat}
Given a differential operator $\Kop$, any composable differential operator
$\Lop$ such that $\Lop \circ \Kop = 0$ is a \emph{compatibility operator} for $\Kop$.
If $\Kop_1$ is a compatibility operator for $\Kop$, it is called \emph{complete}  when any other compatibility operator $\Lop$ can be
factored through $\Lop = \Lop'\circ \Kop_1$ for some differential operator $\Lop'$. A
complex of differential operators $\Kop_l$, $l=0,1,\ldots$ is called a
\emph{compatibility complex} for $\Kop$ when $\Kop_0 = \Kop$ and, for each $l\ge 1$,
$\Kop_l$ is a complete compatibility operator for $\Kop_{l-1}$.
\end{definition}

Logically speaking, what we have defined should be called a
\emph{complete compatibility complex} (a sequence of compatibility
operators, where each is complete), but we follow standard usage where
the adjective \emph{complete} is implied~\cite[Def.1.2.4]{tarkhanov}. It
seems the possible distinction of meanings was not important in the
original literature on overdetermined PDEs.

\begin{definition} \label{def:loc-exact}
Given a (possibly infinite) complex of differential operators $\Kop_l$,
$l=l_{\min}, \ldots, l_{\max}$, we say that it is
\emph{locally exact} at a point $x$ and node $l$ when, for every pair
$(f_l,U)$ of an open neighborhood $U \ni x$ and a smooth section $f_l$
defined on $U$ such that $\Kop_l f_l = 0$, there exists a smooth section
$g_{l-1}$ defined on a possibly smaller open neighborhood $V \ni x$ such
that $f_l = \Kop_{l-1} g_{l-1}$. \emph{Locally exact} without specifying
a point $x$ means locally exact at every $x$, and without specifying $l$
means locally exact for every $l_{\min} < l \le l_{\max}$.
\end{definition}

Our convention is that exactness does not apply at a finite end node of
a complex (initial or final, if they exist). That way, a truncated exact
complex remains exact. Of course, a finite complex can always be
extended by zero maps to any desired length, but that might change its
exactness properties.

\begin{definition} \label{def:homalg}
A (possibly infinite) composable sequence $\Kop_l$ of linear maps,
$l=l_{\min}, \ldots, l_{\max}$, such that
$\Kop_{l+1} \circ \Kop_l = 0$  for each allowed $l$, is called a \emph{(cochain)
complex}. Given complexes $\Kop_l$ and $\Kop'_l$ a sequence $\Cop_l$ of linear
maps, as in the diagram
\begin{equation}
\begin{tikzcd}[column sep=large,row sep=large]
	\cdots \ar{r} \&
	\bullet \ar{r}{\Kop_{l-1}} \ar{d}{\Cop_{l-1}} \&
	\bullet \ar{r}{\Kop_l} \ar{d}{\Cop_l} \&
	\bullet \ar{r}{\Kop_{l+1}} \ar{d}{\Cop_{l+1}} \&
	\bullet \ar{r} \ar{d}{\Cop_{l+2}} \&
	\cdots
	\\
	\cdots \ar{r} \&
	\bullet \ar[swap]{r}{\Kop'_{l-1}} \&
	\bullet \ar[swap]{r}{\Kop'_l} \&
	\bullet \ar[swap]{r}{\Kop'_{l+1}} \&
	\bullet \ar{r} \&
	\cdots
\end{tikzcd} ,
\end{equation}
such that its squares commute, that is $\Kop'_l \circ \Cop_l = \Cop_{l+1} \circ
\Kop_l$ for each allowed $l$, is called a \emph{cochain map} or a \emph{morphism}
between complexes. A \emph{homotopy} between complexes $\Kop_l$ and $\Kop'_l$
(which could also be the same complex, $\Kop_l = \Kop'_l$) is a sequence of
morphisms, as the dashed arrows in the diagram
\begin{equation}
\begin{tikzcd}[column sep=large,row sep=large]
	\cdots \ar{r} \&
	\bullet \ar{r}{\Kop_{l-1}} \ar{d}{\Cop_{l-1}} \&
	\bullet \ar{r}{\Kop_l} \ar{d}{\Cop_l} \ar[dashed]{dl}{\Hop_{l-1}} \&
	\bullet \ar{r}{\Kop_{l+1}} \ar{d}{\Cop_{l+1}} \ar[dashed]{dl}{\Hop_l} \&
	\bullet \ar{r} \ar{d}{\Cop_{l+2}} \ar[dashed]{dl}{\Hop_{l+1}}\&
	\cdots
	\\
	\cdots \ar{r} \&
	\bullet \ar[swap]{r}{\Kop'_{l-1}} \&
	\bullet \ar[swap]{r}{\Kop'_l} \&
	\bullet \ar[swap]{r}{\Kop'_{l+1}} \&
	\bullet \ar{r} \&
	\cdots
\end{tikzcd} .
\end{equation}
The sequence of maps $\Cop_l = \Kop'_{l-1}\circ \Hop_{l-1} + \Hop_l\circ \Kop_l$ is
said to be a \emph{morphism induced by} the homotopy $\Hop_l$. An
\emph{equivalence up to homotopy} between complexes $\Kop_l$ and $\Kop'_l$ is
a pair of morphisms $\Cop_l$ and $\Dop_l$ between them, as in the diagram
\begin{equation}
\begin{tikzcd}[column sep=large,row sep=4.5em]
	\ar[loop left]{}{\tilde{\Hop}_{l_{\min}-1}}
	\bullet \ar{r}{\Kop_{l_{\min}}}
		\ar[swap,shift right]{d}{\Cop_{l_{\min}}} \&
	\ar[r,phantom,"\cdots"]
		\ar[dashed,bend left]{l}{\Hop_{l_{\min}}} \&
	\bullet \ar{r}{\Kop_l}
		\ar[swap, shift right]{d}{\Cop_l} \&
	\bullet \ar[r,phantom,"\cdots"]
		\ar[swap,shift right]{d}{\Cop_{l+1}}
		\ar[dashed,bend left]{l}{\Hop_l} \&
	\ar{r}{\Kop_{l_{\max}}} \&
	\bullet \ar[swap,shift right]{d}{\Cop_{l_{\max}+1}}
		\ar[dashed,bend left]{l}{\Hop_{l_{\max}}}
		\ar[loop right]{}{\tilde{\Hop}_{l_{\max}+1}}
	\\
	\ar[loop left]{}{\tilde{\Hop}'_{l_{\min}-1}}
	\bullet \ar[swap]{r}{\Kop'_{l_{\min}}}
		\ar[swap,shift right]{u}{\Dop_{l_{\min}}} \&
	\ar[r,phantom,"\cdots"]
		\ar[swap,dashed,bend right]{l}{\Hop'_{l_{\min}}} \&
	\bullet \ar[swap]{r}{\Kop'_l}
		\ar[swap,shift right]{u}{\Dop_l} \&
	\bullet \ar[r,phantom,"\cdots"]
		\ar[swap,shift right]{u}{\Dop_{l+1}}
		\ar[swap,dashed,bend right]{l}{\Hop'_l} \&
	\ar[swap]{r}{\Kop'_{l_{\max}}} \&
	\bullet \ar[swap,shift right]{u}{\Dop_{l_{\max}+1}}
		\ar[swap,dashed,bend right]{l}{\Hop'_{l_{\max}}}
	\ar[loop right]{}{\tilde{\Hop}'_{l_{\max}+1}}
\end{tikzcd} ,
\end{equation}
such that $\Cop_l$ and $\Dop_l$ are mutual inverses up to homotopy ($\Hop_l$ and
$\Hop'_l$), that is
\begin{subequations}
\begin{align}
	\Dop_l \circ \Cop_l &= \id - \Kop_{l-1} \circ \Hop_{l-1} - \Hop_l \circ \Kop_l ,
	\\
	\Cop_l \circ \Dop_l &= \id - \Kop'_{l-1} \circ \Hop'_{l-1} - \Hop'_l \circ \Kop'_l ,
\end{align}
\end{subequations}
with the special end cases
\begin{subequations}
\begin{align}
	\Dop_{l_{\min}} \circ \Cop_{l_{\min}}
	&= \id - \tilde{\Hop}_{l_{\min}-1} - \Hop_{l_{\min}} \circ \Kop_{l_{\min}} , &
	\Kop_{l_{\min}} \circ \tilde{\Hop}_{l_{\min}-1} &= 0 ,
	\\
	\Cop_{l_{\min}} \circ \Dop_{l_{\min}}
	&= \id - \tilde{\Hop}'_{l_{\min}-1} - \Hop'_{l_{\min}} \circ \Kop'_{l_{\min}} , &
	\Kop'_{l_{\min}} \circ \tilde{\Hop}'_{l_{\min}-1} &= 0 ,
	\\
	\Dop_{l_{\max}+1} \circ \Cop_{l_{\max}+1}
	&= \id - \Hop_{l_{\max}} \circ \Kop_{l_{\max}} - \tilde{\Hop}_{l_{\max}+1} , &
	\tilde{\Hop}_{l_{\max}+1} \circ \Kop_{l_{\max}}  &= 0 ,
	\\
	\Cop_{l_{\max}+1} \circ \Dop_{l_{\max}+1}
	&= \id - \Hop'_{l_{\max}} \circ \Kop'_{l_{\max}} - \tilde{\Hop}'_{l_{\max}+1} , &
	\tilde{\Hop}'_{l_{\max}+1} \circ \Kop'_{l_{\max}}  &= 0 ,
\end{align}
\end{subequations}
where the $\tilde{\Hop}$ maps are allowed to be arbitrary, as long as they
satisfy the given identities.
\end{definition}

Again, our convention allows the operators constituting a homotopy or an
equivalence up to homotopy between two complexes to satisfy the same
definition when the complexes are truncated.

\begin{definition} \label{def:flat-conn}
A \emph{(linear) connection} $\DD$ on a vector bundle $V\to M$ is a
first order linear differential operator $\DD\colon \Secs(V) \to
\Secs(T^*M \otimes_M V)$ that satisfies the Leibniz rule in the sense
that $\DD(f v) = \d f \otimes v + f \DD v$, for any scalar $f$ and $v\in
\Secs(V)$. Inductively, a \emph{connection} $\DD$ uniquely gives
rise to a sequence of \emph{twisted exterior derivatives}
\begin{equation}
	\d^\DD_l \colon \Secs(\Lambda^l T^*M \otimes_M V)
		\to \Secs(\Lambda^{l+1} T^*M \otimes_M V), \quad l=0,1,\ldots,n,\ldots ,
\end{equation}
with $\d^\DD_0 = \DD$ and degenerating to $\d^\DD_l = 0$ for $l\ge n$,
that satisfy the Leibniz rule in the sense that
$\d^\DD_{l+1}(\alpha\wedge w) = \d\alpha\wedge w - \alpha \wedge \d^\DD_l
w$ for any $1$-form $\alpha$ and $w\in \Secs(\Lambda^l T^*M \otimes_M
V)$. The connection is \emph{flat} when $\d_1^\DD \d^\DD_0 = 0$, 
 in which case the operators $\d^\DD_l$ form a complex called the
\emph{($\DD$-)twisted de~Rham complex}. A section $f$ satisfying  $\DD f = 0$ is said to be 
 \emph{parallel} with respect to $\DD$.\footnote{Parallel sections are also known as \emph{flat sections}.}
\end{definition}

\begin{remark}
The above definition can be made much more explicit if we (locally)
choose coordinates $(x^a)$  and a frame $\bm{e}^i \in \Secs(V)$, so that
an arbitrary section of $v = \nu_i \bm{e}^i \in \Secs(V)$ can be
represented as a linear combination of the frame $\bm{e}^i$ with scalar
coefficients $\nu_i$. Expressing everything in components, 
\begin{equation}
	\DD_a[\nu_i \bm{e}^i]
	= \left[(\DD_a)_i^j \nu_i\right] \bm{e}^i
	= \left[\del_a \nu_i - (\FlatCon_a)_i^j \nu_j\right] \bm{e}^i ,
\end{equation}
where the matrix valued $1$-forms $(\FlatCon_a)_i^j$ are known as the
corresponding \emph{connection coefficients}. The flatness condition
$\d_0^\DD \d_1^\DD = 0$ is then equivalent to the commutativity of $\DD_a$
components as matrix differential operators for different form indices,
$[\DD_a, \DD_b]_i^j = (\DD_a)_i^k (\DD_b)_k^j - (\DD_b)_i^k (\DD_a)_k^j
= 0$. In components, the twisted exterior derivatives act on $\nu_{i,a_1
\cdots a_l} \bm{e}^i \in \Secs(\Lambda^l T^*M \otimes_M V)$ simply as
\begin{equation}
	(\d^\DD_l [\nu_i \bm{e}^i])_{a_1 \cdots a_{l+1}}
	= (l+1) \left[(\DD_{[a_1|})^j_i \nu_{j,|a_2\cdots a_n]}\right] \bm{e}^i .
\end{equation}
\end{remark}

If, as in the above Remark, the frame is chosen to be parallel with respect to $\DD$, $\DD_a
\bm{e}^i=0$, which for a flat connection is always possible locally, then the corresponding
connection coefficients vanish, $(\FlatCon_a)_i^j = 0$. Note though that
the vanishing of the $(\FlatCon_a)_i^j$ is a frame-dependent property,
while the property of being \emph{flat} is completely frame independent.
However, using a flat frame, we can see that the twisted de~Rham complex
is locally equivalent to several copies of the usual de~Rham complex,
and hence, applying appropriate versions of the Poincar\'e lemma to each
copy, we get
\begin{proposition}[{\cite{tarkhanov}\footnote{Propositions 1.2.13, 1.2.39 and 1.2.41}}] \label{prp:de-rham}
Given a flat connection $\DD$, the corresponding twisted de~Rham
complex $\d^\DD_l$, $l=0,1,\ldots,n,\ldots$ is locally exact and is
also a compatibility complex for $\DD = \d^\DD_0$.
\end{proposition}

\begin{proposition}[{\cite[Lem.4]{kh-compat}}] \label{prp:compat-sufficient}
Consider two complexes of differential operators $\Kop_l$ and $\Kop_l'$, for
$l=0,1,\cdots,n-1$, and an equivalence up to homotopy between them, as
in the diagram
\begin{equation}
\begin{tikzcd}[column sep=2cm,row sep=2cm]
	\bullet \ar{r}{\Kop_0}
		\ar[swap,shift right]{d}{\Cop_0} \&
		\ar[dashed,bend left]{l}{\Hop_0}
	\bullet \ar{r}{\Kop_1}
		\ar[swap, shift right]{d}{\Cop_1} \&
	\bullet \ar[r,phantom,"\cdots"]
		\ar[swap,shift right]{d}{\Cop_2}
		\ar[dashed,bend left]{l}{\Hop_1} \&
	\bullet
		\ar[swap,shift right]{d}{\Cop_{n-1}}
		\ar{r}{\Kop_{n-1}} \&
	\bullet \ar[swap,shift right]{d}{\Cop_n}
		\ar[dashed,bend left]{l}{\Hop_{n-1}}
	\\
	\bullet \ar[swap]{r}{\Kop'_0}
		\ar[swap,shift right]{u}{\Dop_0} \&
		\ar[swap,dashed,bend right]{l}{\Hop'_0}
	\bullet \ar[swap]{r}{\Kop'_1}
		\ar[swap,shift right]{u}{\Dop_1} \&
	\bullet \ar[r,phantom,"\cdots"]
		\ar[swap,shift right]{u}{\Dop_2}
		\ar[swap,dashed,bend right]{l}{\Hop'_1} \&
	\bullet
		\ar[swap,shift right]{u}{\Dop_{n-1}}
		\ar[swap]{r}{\Kop'_{n-1}} \&
	\bullet \ar[swap,shift right]{u}{\Dop_n}
		\ar[swap,dashed,bend right]{l}{\Hop'_{n-1}}
\end{tikzcd} ,
\end{equation}
where for simplicity we are assuming that the $\tilde{\Hop}_{-1}$, $\tilde{\Hop}'_{-1}$, $\tilde{\Hop}_n$, $\tilde{\Hop}'_n$ are all zero.
\begin{enumerate}[label=\alph*)] 
 \item If $\Kop'_l$ is a compatibility complex for $\Kop'_0$, then $\Kop_l$ is a compatibility complex for $\Kop_0$. 
 \item If $\Kop'_l$ is locally exact, then $\Kop_l$ is locally exact.
\end{enumerate}  
\end{proposition}

\section{Preliminaries and notation} \label{sec:prel}  
Unless otherwise stated we work in a Lorentzian 4-dimensional spacetime of signature ${+}{-}{-}{-}$, using the 2-spinor formalism following the notations and conventions of \cite{Penrose:1986fk}. In particular, indices are lowered and raised with $\epsilon_{AB} = - \epsilon_{BA}$ and its inverse according to the rules
\begin{align} \label{eq:epsilonRasingLowering}
\kappa_B = \epsilon_{AB}\kappa^A, &&
\kappa^A = \epsilon^{AB} \kappa_B.
\end{align}
Analogously, on the primed spin space, indices are shifted with $\bar{\epsilon}_{A'B'}$ and its inverse. These isomorphisms will be used throughout this section, for example to identify vectors with 1-forms. The identity map on a vector space $V_k$ is denoted $\id_k$.

\subsection{Spinor calculus} \label{sec:spinor}
We make use of the fact that any tensor can be decomposed into a set of symmetric spinors. Let $\SymSpin_{k,l}$ be the space of symmetric valence $(k,l)$ spinors. In abstract index notation, elements are of the form $\phi_{A_1 \dots A_k A'_1 \dots A'_l} \in \SymSpin_{k,l}$. Sometimes it is convenient to suppress the valence and/or indices and we write e.g. $\phi \in \mathcal{S}$ or $\phi \in \mathcal{S}_{k,l}$. Furthermore, for collections of symmetric spinors we use the shorthand $\SymSpin_{k,l}=(k,l)$. For example the decomposition of a symmetric 2-tensor is an element of
\begin{align}
 \begin{bmatrix} (2,2) \\ (0,0) \end{bmatrix},
\end{align}
with the trace-free symmetric part as first element and the trace as second element. For readers less familiar with the spinor formalism, it is worth noting
that for spinors with even numbers of indices it is possible to identify
each of the spinor spaces used in this paper with complexified tensor spaces
according to the following table: 
\begin{center}
\begin{tabular}{rccccccc}
	\\[0.5ex]
	\hline\hline
	\\[-2ex]
	spinor & $(0,0)$ & $(1,1)$
		& $(2,2)$
		& $(2,0)$ & $(0,2)$
		& $(4,0)$ & $(0,4)$
	\\[2ex]
	tensor
		& $\mathbb{C},\yd{1,1,1,1}$
		& $\yd{1},\yd{1,1,1}$
		& $\yd{2}_0$
		& $\yd{1,1} + i{*}\yd{1,1}$ & $\yd{1,1} - i{*}\yd{1,1}$
		& $\yd{2,2}_0 + i{*}\yd{2,2}_0$ & $\yd{2,2}_0 - i{*}\yd{2,2}_0$
	\\[2.5ex]
	& \parbox{3em}{\centering scalar,\\ 4-form}
		& \parbox{3em}{\centering vector,\\ 3-form}
		& \parbox{4.5em}{\centering symmetric traceless 2-tensor}
		& \parbox{4em}{\centering anti-self-dual 2-form}
		& \parbox{4em}{\centering self-dual 2-form}
		& \parbox{4em}{\centering anti-self-dual Weyl}
		& \parbox{4em}{\centering self-dual Weyl}
	\\[3.25ex]
	\hline\hline
	\\[0.5ex]
\end{tabular}
\end{center}
Recall that Young diagrams, which we have used in the above table,
represent tensor spaces obtained by filling each box with an index, then
symmetrizing along the rows and finally antisymmetrizing along the
columns~\cite[vol.I, p.143]{Penrose:1986fk}. The $0$ subscript further
denotes the traceless part, while $*$ denotes the Hodge dual against the
antisymmetric indices from the first column of the diagram. The usual
4-dimensional identifications by Hodge duality are also included.

It is convenient to introduce the following product between symmetric spinors.
\begin{definition}
Let $i,j,k,l,m,n$ be integers with $i \leq min(k,m)$ and $j \leq min(l,n)$. The symmetric product is a bilinear form 
\begin{align}
\overset{i,j}{\odot}: \SymSpin_{k,l} \times \SymSpin_{m,n} \to{}& \SymSpin_{k+m-2i,l+n-2j}.
\end{align}
For $\phi \in \SymSpin_{k,l}, \psi \in \SymSpin_{m,n}$ it is given by
\begin{align}
(\phi\overset{i,j}{\odot}\psi)_{A_1 \dots A_{k+m-2i}}^{A'_1 \dots A'_{l+n-2j}}={}& 
\phi_{(A_1 \dots A_{k-i-1}}^{(A'_1 \dots A'_{l-j-1}| B_1 \dots B_i B'_1 \dots B'_j|} 
\psi^{A'_{l-j} \dots A'_{l+n-2j})}_{A_{k-i} \dots A_{k+m-2i})B_1 \dots B_i B'_1 \dots B'_j}.
\end{align}
\end{definition} 
This operation involves $i$ contractions with $\epsilon_{AB}$ (and/or its inverse) and $j$ contractions with $\bar{\epsilon}_{A'B'}$ (and/or its inverse) as indicated in \eqref{eq:epsilonRasingLowering}.

\begin{definition}[\protect{\cite[\S 2.1]{ABB:symop:2014CQGra..31m5015A}}]
The four \emph{fundamental spinor operators} are the differential operators
$$
\sDiv:\mathcal{S}_{k,l}\rightarrow \SymSpin_{k-1,l-1}, \quad 
\sCurl:\mathcal{S}_{k,l}\rightarrow \SymSpin_{k+1,l-1}, \quad 
\sCurlDagger:\mathcal{S}_{k,l}\rightarrow \SymSpin_{k-1,l+1}, \quad 
\sTwist:\mathcal{S}_{k,l}\rightarrow \SymSpin_{k+1,l+1}
$$
defined via
\begin{align} 
\sDiv \varphi ={}& (\nabla \overset{1,1}{\odot}\varphi), &&&
\sCurl \varphi ={}& (\nabla \overset{0,1}{\odot}\varphi), &&&
\sCurlDagger \varphi ={}& (\nabla \overset{1,0}{\odot}\varphi), &&&
\sTwist \varphi ={}& (\nabla \overset{0,0}{\odot}\varphi).
\label{eq:FundamentalOperators}
\end{align}
The operators are called respectively the divergence, curl, curl-dagger, and twistor operators. 
\end{definition}
The irreducible decomposition of a covariant derivative of a symmetric spinor $\varphi \in \SymSpin_{k,l}$ can be written as
\begin{align}
\nabla_A{}^{A'} \varphi_{A_1 \cdots A_k}{}^{A'_1 \cdots A'_l} ={}& 
(\sTwist \varphi)_{AA_1 \cdots A_k}{}^{A'A'_1 \cdots A'_l} - \tfrac{l}{l+1} \bar{\epsilon}^{A'(A'_1}(\sCurl \varphi)_{AA_1 \cdots A_k}{}^{A'_2 \cdots A'_l)} \nonumber \\
&\hspace{-15ex}- \tfrac{k}{k+1}\epsilon_{A(A_1}(\sCurlDagger \varphi)_{A_2 \cdots A_k)}{}^{A'A'_1 \cdots A'_l} + \tfrac{kl}{(k+1)(l+1)}\epsilon_{A(A_1}\bar{\epsilon}^{A'(A'_1}(\sDiv \varphi)_{A_2 \cdots A_k)}{}^{A'_2 \cdots A'_l)}.
\end{align}
Note that in contrast to \cite{ABB:symop:2014CQGra..31m5015A} we suppress valence indices on the operators. With respect to complex conjugation, the operators satisfy $\overline{\sDiv} = \sDiv$, $\overline{\sTwist} = \sTwist$, $\overline{\sCurl} = \sCurlDagger$, $\overline{\sCurlDagger} = \sCurl$, but note that $\overline{\SymSpin_{k,l}}=\SymSpin_{l,k}$. Commutation formulas for the fundamental operators are given in \cite[\S 2.2]{ABB:symop:2014CQGra..31m5015A}. 

In this notation, the Weyl spinor $\Psi\in \SymSpin_{4,0}$, the trace-free Ricci spinor $\Phi\in \SymSpin_{2,2}$ and the Ricci scalar $\Lambda\in \SymSpin_{0,0}$ are related by the Bianchi identity 
\begin{align}
\sCurlDagger \Psi={}& \sCurl \Phi, &
\sDiv \Phi ={}& -3\sTwist \Lambda. \label{eq:Bianchi}
\end{align}

\subsection{Kerr geometry}
The main feature of the Kerr geometry is encoded in the Killing spinor $\kappa\in \SymSpin_{2,0}$ found in \cite{walker:penrose:1970CMaPh..18..265W}, satisfying 
\begin{align} \label{eq:Twistkappa}
\sTwist \kappa = 0 
\end{align} 
In a principal dyad the Killing spinor takes the simple form 
\begin{align} \label{eq:TypeDKS}
\kappa_{AB}={}&-2 \kappa_1 o_{(A}\iota_{B)},
\end{align}
with $\kappa_1 \propto \Psi_2^{-1/3}$ and $\Psi_2$ being the only non-vanishing component of the Weyl spinor. Note that $\kappa_1$ and $\Psi_2$ can be expressed covariantly via the relations 
\begin{align}
\kappa_{AB} \kappa^{AB} &= -2 \kappa_{1}{}^2,&
\Psi_{ABCD} \Psi^{ABCD}&=6 \Psi_{2}^2.
\label{eq:Psi2Covariant}
\end{align}
Hence, we can allow $\kappa_1$ and $\Psi_2$ in covariant expressions.  The tensor product of $\kappa_{AB}$ with a  symmetric spinor has at most three different irreducible components. These involve either zero, one or two contractions and symmetrization. For these operations we introduce the $\mathcal{K}$-operators (c.f. \cite[Definition II.4]{2016arXiv160106084A})
\begin{definition}
Given the Killing spinor \eqref{eq:TypeDKS}, define the operators $\mathcal{K}^i:\SymSpin_{k,l}\rightarrow \SymSpin_{k-2i+2,l}, i=0,1,2$  via
\begin{align}
(\mathcal{K}^0 \varphi)\defeq{}&2\kappa_1^{-1} (\kappa \overset{0,0}{\odot}\varphi), &&&
(\mathcal{K}^1 \varphi)\defeq{}&
\kappa_1^{-1}(\kappa  \overset{1,0}{\odot}\varphi), &&&
(\mathcal{K}^2 \varphi)\defeq{}&- \tfrac{1}{2}\kappa_1^{-1}(\kappa  \overset{2,0}{\odot} \varphi).
 \label{eq:Kprojectors}
\end{align}
\end{definition}
Note that the complex conjugated operators act on the primed indices in the analogous way. To compare to results in the literature in section~\ref{sec:equiv} we define an algebraic projection operator on $\SymSpin_{4,0}$ by
\begin{align}\label{eq:P2Def}
\mathcal{P}^{2}\defeq{}& \mathcal{K}^1 \mathcal{K}^1 \mathcal{K}^1 \mathcal{K}^1  - \tfrac{1}{16}  \mathcal{K}^0 \mathcal{K}^1 \mathcal{K}^1 \mathcal{K}^2, 
\end{align}
cf. \cite{2016arXiv160106084A}. In a principal dyad, $\mathcal{P}^{2} \varphi$ has components $(\varphi_0, 0, 0, 0, \varphi_4)$. 
It follows from \eqref{eq:Twistkappa} that
\begin{align}\label{eq:XiDef}
\xi_{AA'}\defeq{}&(\sCurlDagger_{2,0} \kappa)_{AA'}, 
\end{align}
is a Killing vector field, which is real after suitably normalizing $\kappa_{AB}$. The second Killing vector follows from contraction with the Killing spinor and its complex conjugate, or equivalently a Killing tensor, via
\begin{align}
\zeta_{AA'}\defeq{}&\tfrac{9}{2} \bigl(\kappa_{AB} \bar{\kappa}_{A'B'} + \tfrac{1}{4} \epsilon_{AB} \bar\epsilon_{A'B'} (\kappa_{CD} \kappa^{CD} + \bar{\kappa}_{C'D'} \bar{\kappa}^{C'D'})\bigr) \xi^{BB'}.
\end{align}
Another important vector field is defined by
\begin{align} \label{eq:UU11Def}
U_{AA'}\defeq{}&- \frac{\kappa_{AB} \xi^{B}{}_{A'}}{3 \kappa_1^2} = - \nabla_{AA'}\log(\kappa_1) .
\end{align}
Although our treatment is independent of coordinate choices, we get $\kappa_1=-(r-ia\cos\theta)/3$, $\xi=\partial_t$  and $\zeta=a^2 \partial_t + a \partial_\phi$ in standard Boyer-Lindquist coordinates. Note that $\xi^a \zeta_a \neq 0$.

As an example of the irreducible spinor notation used in this paper, the de Rham complex is given by 
\begin{gather} \label{eq:deRhamComplex}
\begin{tikzcd}[column sep=10ex]
	\begin{bmatrix} (0,0) \end{bmatrix}
		\ar{r}{\begin{bmatrix} \sTwist \end{bmatrix}}	\&
	\begin{bmatrix} (1,1) \end{bmatrix}
		\ar{r}{\begin{bmatrix} \sCurlDagger\\ \sCurl \end{bmatrix}}	\&
	\begin{bmatrix} (0,2)\\ (2,0) \end{bmatrix}
		\ar{r}{\begin{bmatrix} \sCurl & -\sCurlDagger  \end{bmatrix}}	\&
	\begin{bmatrix} (1,1) \end{bmatrix}
		\ar{r}{\begin{bmatrix} \sDiv  \end{bmatrix}}	\&
	\begin{bmatrix} (0,0) \end{bmatrix}
\end{tikzcd}.
\end{gather}
Here 3-forms are identified with vectors via Hodge duality, see the table at the beginning of this section for more details about the spinor representation.

\subsection{Linearized curvature}
We will describe the linearized curvature operators in spinor form using the operator $\vartheta$ introduced in 
\cite{Backdahl:2015yua}. On a Kerr background the linearized curvature spinors 
\begin{align}
\begin{bmatrix}
\vartheta \Lambda \\
\vartheta \Phi \\
\vartheta \Psi 
\end{bmatrix} : \begin{bmatrix} (2,2) \\ (0,0) \end{bmatrix} \rightarrow \begin{bmatrix} (0,0)\\ (2,2) \\ (4,0) \end{bmatrix}
\end{align}
take the form
\begin{align} \label{eq:LinCurvOps}
\begin{bmatrix}
\vartheta \Lambda \\
\vartheta \Phi \\
\vartheta \Psi 
\end{bmatrix}={}&\begin{bmatrix}
- \tfrac{1}{24}\sDiv \sDiv & \tfrac{1}{32}\sDiv \sTwist \\
 \tfrac{1}{6}\sTwist \sDiv + \tfrac{1}{2}\sCurlDagger \sCurl + \tfrac{1}{2}\Psi_{2} - \tfrac{3}{4}\Psi_{2}\mathcal{K}^0 \mathcal{K}^2 & - \tfrac{1}{8}\sTwist \sTwist \\
\tfrac{1}{2}\sCurl \sCurl & - \tfrac{3}{32}\Psi_{2}\mathcal{K}^0 \mathcal{K}^0
\end{bmatrix}.
\end{align}

\section{Killing compatibility complex on Kerr} \label{sec:kerr}
In this section we present our main results, which constitute a proof of
the completeness of the set of local gauge invariants on Kerr spacetime,
which were introduced in~\cite{ab-kerr}. In
addition, we construct the full compatibility complex $\Kop_l$,
$l=0,1,2,3$, for the Killing operator $\Kop = \Kop_0$ on
Kerr. That is, the components of
$\Kop_l$ identify a complete list of differential relations
between the components of $\Kop_{l-1}$. As will be shown in Section~\ref{sec:equiv}, the operator $\Kop_1$ and the invariants from~\cite{ab-kerr} factor through each other, thus confirming their completeness.

The proof relies in an essential way on the material reviewed in
Sections~\ref{sec:compat} and~\ref{sec:prel}. Namely, the fundamental
spinorial objects and differential operators used to give explicit formulas for the
compatibility complex $\Kop_l$, as well as its equivalence up to
homotopy to an auxiliary complex $\Kop'_l = \d_l^\DD$,
\begin{equation} \label{diag:DiffComplex4D}
\begin{tikzcd}
 V_0 \ar[swap,shift right]{d}{\Cop_0} \arrow[r, "\Kop_0"]  \& 
 V_1 \ar[swap,shift right]{d}{\Cop_1} \arrow[r, "\Kop_1"] \arrow[l, dashed, bend left=20, "\Hop_0"]  \& 
 V_2 \ar[swap,shift right]{d}{\Cop_2} \arrow[r, "\Kop_2"] \arrow[l, dashed, bend left=20, "\Hop_1"] \& 
 V_3 \ar[swap,shift right]{d}{\Cop_3} \arrow[r, "\Kop_3"] \arrow[l, dashed, bend left=20, "\Hop_2"] \& 
 V_4 \ar[swap,shift right]{d}{\Cop_4} \arrow[l, dashed, bend left=20, "\Hop_3"]
 \\
V'_0 \ar[swap,shift right]{u}{\Dop_0} \arrow[r, "\Kop'_0"'] \& 
V'_1 \ar[swap,shift right]{u}{\Dop_1} \arrow[r, "\Kop'_1"'] \arrow[l, dashed, bend right=20, "\Hop'_0"'] \& 
V'_2 \ar[swap,shift right]{u}{\Dop_2} \arrow[r, "\Kop'_2"'] \arrow[l, dashed, bend right=20, "\Hop'_1"'] \& 
V'_3 \ar[swap,shift right]{u}{\Dop_3} \arrow[r, "\Kop'_3"'] \arrow[l, dashed, bend right=20, "\Hop'_2"'] \& 
V'_4 \ar[swap,shift right]{u}{\Dop_4} \arrow[l, dashed, bend right=20, "\Hop'_3"']
\end{tikzcd} ,
\end{equation} 
which happens to be the twisted de~Rham complex corresponding to the
unique flat connection $\DD$ defined on the 2-dimensional sub-bundle $V'_0
\hookrightarrow V_0 = T^*M$ spanned by local solutions of the Killing
equation $\Kop[\nu] = 0$ such that the Killing vectors themselves
are parallel with respect to $\DD$. If we choose to write an arbitrary section of $V'_0$ as
$v = \alpha \xi + \beta \zeta$, where $\alpha$, $\beta$ are arbitrary
scalar functions and the 1-forms $\xi$, $\zeta$ constitute a basis of the
$2$-dimensional space of solutions of the Killing equation on Kerr, then
this flat connection simply acts as
\begin{equation} \label{eq:DDdef}
	\DD[v] = (\d\alpha) \xi + (\d\beta) \zeta ,
\end{equation}
where $\d$ is the ordinary exterior derivative. That is, if we
choose to parametrize the $V'_0$ bundle using the $(\xi,\zeta)$-frame,
then the connection coefficients of $\DD$ vanish identically and the complex $\Kop'_l$ simply corresponds to the direct sum of
two copies of the ordinary de~Rham complex \eqref{eq:deRhamComplex}. However here we choose the frame 
\begin{align} \label{eq:frame}
e_{1}{}^a \defeq{}& \xi^a, &&&
e_{2}{}^a \defeq{}& \zeta^a,&&&
e_{3}{}^a \defeq{}& \kappa_{1}{}^2 \bar{\kappa}_{1'}{}U{}^a,&&&
e_{4}{}^a \defeq{}&\overline{e_3{}^a},
\end{align}
for $V_0$ with the co-frame $e^i{}_a$ defined to satisfy $e^i{}_a  e_j{}^a = \delta^i_j$. Note that $e_1$ and $e_2$ are orthogonal to $e_3$ and $e_4$, as can be seen from \eqref{eq:UU11Def} in Boyer-Lindquist coordinates, and we use $e^1$ and $e^2$ as a frame for $V'_0$.\footnote{By definition, $e^1(\xi) = 1 = e^2(\zeta)$ and $e^1(\zeta)= 0 = e^2(\xi)$, so the difference from the frame used in  \eqref{eq:DDdef} is essentially due to the non-orthogonality of $\xi$ and $\zeta$.} In this 2-dimensional frame, the connection coefficients of $\DD$ are non-vanishing, but it turns out to be preferred for computations in section~\ref{sec:C2tildeComps}. The twisted connection is given by the four real 1-forms
\begin{align}
\FlatCon^{i}{}_{j}{}\defeq{}&(-2e_{j}{}{\overset{1,1}{\odot}}\sTwist e^{i}{}), \qquad i,j=1,2,
\end{align}
and the twisted de Rham complex is given by two copies of the de Rham complex \eqref{eq:deRhamComplex}, but with non-vanishing connection, via
\begin{subequations} \label{eq:Kpops} 
\begin{align}
\Kop'_0\defeq{}&\begin{bmatrix}
\sTwist\bullet - \FlatCon^{1}{}_{1}{}{\overset{0,0}{\odot}}\bullet & - \FlatCon^{2}{}_{1}{}{\overset{0,0}{\odot}}\bullet\\
- \FlatCon^{1}{}_{2}{}{\overset{0,0}{\odot}}\bullet & \sTwist\bullet - \FlatCon^{2}{}_{2}{}{\overset{0,0}{\odot}}\bullet
\end{bmatrix}, \\
\Kop'_1\defeq{}&\begin{bmatrix}
\sCurlDagger\bullet - \FlatCon^{1}{}_{1}{}{\overset{1,0}{\odot}}\bullet & - \FlatCon^{2}{}_{1}{}{\overset{1,0}{\odot}}\bullet\\
\sCurl\bullet - \FlatCon^{1}{}_{1}{}{\overset{0,1}{\odot}}\bullet & - \FlatCon^{2}{}_{1}{}{\overset{0,1}{\odot}}\bullet\\
- \FlatCon^{1}{}_{2}{}{\overset{1,0}{\odot}}\bullet & \sCurlDagger\bullet - \FlatCon^{2}{}_{2}{}{\overset{1,0}{\odot}}\bullet\\
- \FlatCon^{1}{}_{2}{}{\overset{0,1}{\odot}}\bullet & \sCurl\bullet - \FlatCon^{2}{}_{2}{}{\overset{0,1}{\odot}}\bullet
\end{bmatrix}, \\
\Kop'_2\defeq{}&\begin{bmatrix}
\sCurl\bullet - \FlatCon^{1}{}_{1}{}{\overset{0,1}{\odot}}\bullet & - \sCurlDagger\bullet + \FlatCon^{1}{}_{1}{}{\overset{1,0}{\odot}}\bullet & - \FlatCon^{2}{}_{1}{}{\overset{0,1}{\odot}}\bullet & \FlatCon^{2}{}_{1}{}{\overset{1,0}{\odot}}\bullet\\
- \FlatCon^{1}{}_{2}{}{\overset{0,1}{\odot}}\bullet & \FlatCon^{1}{}_{2}{}{\overset{1,0}{\odot}}\bullet & \sCurl\bullet - \FlatCon^{2}{}_{2}{}{\overset{0,1}{\odot}}\bullet & - \sCurlDagger\bullet + \FlatCon^{2}{}_{2}{}{\overset{1,0}{\odot}}\bullet
\end{bmatrix}, \\
\Kop'_3\defeq{}&\begin{bmatrix}
\sDiv\bullet - \FlatCon^{1}{}_{1}{}{\overset{1,1}{\odot}}\bullet & - \FlatCon^{2}{}_{1}{}{\overset{1,1}{\odot}}\bullet\\
- \FlatCon^{1}{}_{2}{}{\overset{1,1}{\odot}}\bullet & \sDiv\bullet - \FlatCon^{2}{}_{2}{}{\overset{1,1}{\odot}}\bullet
\end{bmatrix}.
\end{align}
\end{subequations}
The flatness of the connection, as evident in \eqref{eq:DDdef}, is equivalent to
\begin{align} \label{eq:FlatConnectionConditions}
\sCurl \FlatCon^i{}_j = - \sum_{k=1}^2 \FlatCon^i{}_k {\overset{1,0}{\odot}} \FlatCon^k{}_j, &&
\sDiv \FlatCon^i{}_j = 0
\end{align}
and the complex conjugate relations.
\begin{remark} \label{rmk:intuition}
The construction of the complex $\Kop_l$ is heavily patterned on the
analogous construction for the Schwarzschild spacetime carried out
in~\cite[Sec.3.3]{kh-compat}. As such, we do not reproduce a fully
detailed discussion of the construction, but only give the final result
and enough information to show that the construction is correct, namely
that all the identities implicit in the
diagram~\eqref{diag:DiffComplex4D} hold true. However, we can briefly
summarize the intuition behind the construction. The resulting $\Kop_1$
operator consists of two groups of invariants, $\Kop'_1\Cop_1$ and
$\Pop^1_\top (\id_1 - \Kop_0 \Hop_0)$, evident in the
notation of~\eqref{eq:1aops} below. The construction of each group mimics a
well-known pattern, both of which can be conveniently found
in the work of Martel \& Poisson~\cite{martel-poisson}, which reviews the construction of mode-level
gauge invariants on the Schwarzschild background.

The pattern for $\Pop^1_\top (\id_1 - \Kop_0 \Hop_0)$ is as follows.
In~\cite{martel-poisson}, whose equations we will prefix by MP for clarity,
after mode decomposition, Equations~(MP4.6--9) show
the explicit gauge transformations of the various even metric components
$h_{ab}$, $j_a$, $K$ and $G$. Then~(MP4.12) identifies $\varepsilon_a$ as
a differential operator on metric components that transforms exactly by
the vector part $\xi_a$ of the gauge parameters; for us that operator is
$\Hop_0$. The gauge invariant variables $\tilde{h}_{ab}$~(MP4.10) 
are then constructed by subtracting from the
corresponding metric components a differential operator acting on
$\varepsilon_a$ to exactly cancel their transformation by the gauge
parameter $\xi_a$; for us this subtraction takes the form $\id - \Kop_0
\Hop_0$, with the projection $\Pop^1_\top$ picking out precisely those
metric components for which the cancellation of the gauge parameter
dependence is complete. The same pattern is explicitly recognized in the
construction of the gauge invariant scalars $\{\mathcal{I}_1,
\mathcal{I}_2, \mathcal{I}_3\}$ in Section~II of~\cite{mobpm-kerr}.

The pattern for $\Kop'_1 \Cop_1$ also appears in~\cite{martel-poisson},
but somewhat implicitly. After mode decomposition, Equations~(MP5.5--6)
show the explicit gauge transformations of the odd metric
components $h_a$ and $h_2$, where the former can be rewritten
$h'_a/r^2 = h_a/r^2 - \nabla_a (\xi/r^2)$, where $r$ is the standard Schwarzschild radial coordinate. 
What is crucial here is that the dependence on the gauge parameter $\xi$ appears
through the \emph{gradient} of $\xi/r^2$; for us the analogous
identity is $\Cop_1 \Kop_0 = \Kop'_0 \Cop_0$, where $\Cop_0$ is
analogous to the rescaling of the gauge parameter $\xi$, $\Cop_1$ is
analogous to the projection onto the rescaled metric components
$h_a/r^2$, and $\Kop'_0$ is analogous to the gradient. Then the
unnumbered formula in Section~V.C of~\cite{martel-poisson} shows that
the gauge invariant \emph{Cunningham-Price-Moncrief} scalar can be
defined as $\Psi_{\text{odd}} \sim  \varepsilon^{ab}\nabla_{a}
(h_{b}/r^2)$. What is crucial here is the appearance of the curl of
$h_a/r^2$, which precisely kills the gauge transformation of
$h_a/r^2$ by the gradient of $\xi/r^2$, which for us is
analogous to the composition $\Kop'_1 \Cop_1$, where $\Kop'_1$ plays the
role of the higher dimensional curl.
\end{remark}

If one strips away all the layers of abstraction from the results of
Section~\ref{sec:compat}, the remaining core result is that a judicious
application of the two patterns from the remark above is sufficient to
construct a complete set of linear gauge invariant observables (on
geometries where the number of independent Killing vectors is locally
constant).

Returning to the construction, the differential operators $\Kop_l$ and $\Kop'_l$ act between functions
valued in the vector spaces $\Vec_i$ and $\Vec'_i$, which are composed
of symmetric spinors as follows
\begin{subequations} 
\begin{align}
\Vec_0 &\defeq \begin{bmatrix} (1,1) \end{bmatrix}, &
\Vec_1 &\defeq \begin{bmatrix} (2,2) \\ (0,0) \end{bmatrix}, &
\Vec_2 &\defeq \begin{bmatrix} \Vec'_2 \\ (0,0) \\ (0,0) \\ (0,0) \end{bmatrix}, &
\Vec_3 &\defeq \begin{bmatrix} \Vec'_3 \\ (0,0) \\(0,0) \\ (0,0) \end{bmatrix}, &
\Vec_4 &\defeq \Vec'_4,
\\
\Vec'_0 &\defeq \begin{bmatrix} (0,0) \\ (0,0) \end{bmatrix}, &
\Vec'_1 &\defeq \begin{bmatrix} (1,1) \\ (1,1) \end{bmatrix}, &
\Vec'_2 &\defeq \begin{bmatrix} (0,2) \\ (2,0) \\ (0,2) \\(2,0)  \end{bmatrix}, &
\Vec'_3 &\defeq \begin{bmatrix} (1,1) \\ (1,1) \end{bmatrix}, &
\Vec'_4 &\defeq \begin{bmatrix} (0,0) \\ (0,0) \end{bmatrix}.
\end{align}
\end{subequations}

Due to the geometry of Kerr spacetime there are certain distinguished subspaces of $\Vec_i, \Vec'_i$.
These subspaces are analogous to those identified in the
construction of the compatibility complex on
Schwarzschild~\cite[Sec.3.3]{kh-compat}.
Before we discuss the differential operators for the complex, let us describe subspaces for $\Vec_0, \Vec_1, \Vec'_1$. $\Vec_0$ naturally decomposes into the two dimensional space $\Vec_{\parallel}$, spanned by the Killing vectors $\xi, \zeta$, and its orthogonal complement $\Vec_\perp$, with corresponding mappings
\begin{gather}
\begin{tikzcd}
	V_0 = \begin{bmatrix} (1,1) \end{bmatrix}
		\ar[shift left]{r}{\Pop^0_\parallel}
	\&
	\begin{bmatrix} (0,0) \\ (0,0) \end{bmatrix} = V_\parallel \cong V'_0
		\ar[shift left]{l}{\Pop^\parallel_0}
\end{tikzcd} ,
	\\
\begin{tikzcd}
	V_0 = \begin{bmatrix} (1,1) \end{bmatrix}
		\ar[shift left]{r}{\Pop^0_\perp}
	\&
	\begin{bmatrix} (0,0) \\ (0,0) \end{bmatrix} = V_\perp
		\ar[shift left]{l}{\Pop^\perp_0}
\end{tikzcd}.
\end{gather}
The projection operators in the frame \eqref{eq:frame} take the explicit form
\begin{align}
\label{eq:PopDef}
\Pop^0_{\parallel}  \defeq{}& \begin{bmatrix} 
e_{1}{} \overset{1,1}{\odot}\bullet\\
e_{2}{} \overset{1,1}{\odot}\bullet
\end{bmatrix}, 
&
\Pop^{\parallel}_0  \defeq{}& \begin{bmatrix} 
e^{1}{} \overset{0,0}{\odot}\bullet &
e^{2}{} \overset{0,0}{\odot}\bullet
\end{bmatrix}, 
&
\Pop^0_\perp  \defeq{}& \begin{bmatrix} 
e_{3}{} \overset{1,1}{\odot}\bullet\\
e_{4}{} \overset{1,1}{\odot}\bullet
\end{bmatrix}, 
&
\Pop^\perp_0  \defeq{}& \begin{bmatrix} 
e^{3}{} \overset{0,0}{\odot}\bullet &
e^{4}{} \overset{0,0}{\odot}\bullet
\end{bmatrix},
\end{align}
and satisfy
\begin{align}\label{eq:V0decomp}
\Pop^0_{\parallel} \Pop^{\parallel}_0 ={}& \id_{\parallel},
&
\Pop^0_\perp \Pop^\perp_0 ={}& \id_{\perp},
&
\Pop^0_{\parallel} \Pop^\perp_0 ={}& 0,
&
\Pop^0_\perp \Pop^{\parallel}_0 ={}& 0,
&
\Pop^{\parallel}_0 \Pop^0_{\parallel} + \Pop^\perp_0 \Pop^0_\perp ={}& \id_0.
\end{align}
The space of symmetric 2-tensors, $\Vec_1$, incorporates a 3-dimensional subspace $\Vec_\top$ spanned by the $33$, $34$ and $44$ frame components. It is distinguished due to the non-trivial background curvature $\Psi_2$ and characterized by the maps
\begin{gather}
\begin{tikzcd}
	V_1 = \begin{bmatrix} (2,2) \\ (0,0) \end{bmatrix}
		\ar[shift left]{r}{\Pop^1_\top}
	\&
	\begin{bmatrix} (0,0) \\ (0,0) \\ (0,0) \end{bmatrix} = V_\top
		\ar[shift left]{l}{\Pop^\top_1}
\end{tikzcd} ,
\end{gather}
where
\begin{align}
\Pop^1_\top \defeq{}& \begin{bmatrix}
\Pop^1_{33} \\ \Pop^1_{34} \\ \Pop^1_{44}
\end{bmatrix}, \qquad \text{ with }  \qquad
\Pop^1_{ij} \defeq{}\begin{bmatrix} 
e_{i}{} \overset{1,1}{\odot}e_{j}{} \overset{1,1}{\odot}\bullet & \tfrac{1}{4}e_{i}{} \overset{1,1}{\odot}e_{j}{} \overset{0,0}{\odot}\bullet \end{bmatrix},\qquad i,j=3,4 \\
\Pop^{\top}_1 \defeq{}& \begin{bmatrix}
e_{3}{} \overset{0,0}{\odot}e_{3}{} \overset{0,0}{\odot}\bullet &
2e_{3}{} \overset{0,0}{\odot}e_{4}{} \overset{0,0}{\odot}\bullet &
e_{4}{} \overset{0,0}{\odot}e_{4}{} \overset{0,0}{\odot}\bullet \\
e_{3}{} \overset{1,1}{\odot}e_{3}{} \overset{0,0}{\odot}\bullet &
2e_{3}{} \overset{1,1}{\odot}e_{4}{} \overset{0,0}{\odot}\bullet &
e_{4}{} \overset{1,1}{\odot}e_{4}{} \overset{0,0}{\odot}\bullet 
\end{bmatrix}. 
\end{align} 
They satisfy
\begin{align}
\Pop^1_\top \Pop^{\top}_1 = \id_\top.
\end{align}
A 1-dimensional subspace $\Vec_a$ of $\Vec'_1$ is defined by the anti-symmetric $1,2$ component of the product of the two vector representations and the corresponding maps are given by
\begin{gather*}
\begin{tikzcd}
	V'_1 = \begin{bmatrix} (1,1) \\ (1,1) \end{bmatrix}
		\ar[shift left]{r}{\Pop^{1'}_a}
	\&
	\begin{bmatrix} (0,0) \end{bmatrix} = V_a
		\ar[shift left]{l}{\Pop^a_{1'}}
\end{tikzcd} ,
\end{gather*}
where
\begin{align} 
\Pop^a_{1'}  \defeq{}& \begin{bmatrix} 
e^{2}{} \overset{0,0}{\odot}\bullet\\
-e^{1}{} \overset{0,0}{\odot}\bullet
\end{bmatrix}, &
\Pop^{1'}_a \defeq{}& \frac{1}{2}\begin{bmatrix} 
e_{2}{} \overset{1,1}{\odot}\bullet &
- e_{1}{} \overset{1,1}{\odot}\bullet
\end{bmatrix}. \label{eq:Pop1paDef}
\end{align}
They satisfy
\begin{align} \label{eq:Ida}
\Pop^{1'}_{a} \Pop^a_{1'} ={}&\id_a .
\end{align}
A 7-dimensional subspace $\Vec_s \subset \Vec'_1$ is defined as the image of $\Pop^1_s$ given by
\begin{gather}
\begin{tikzcd}
	V_1 = \begin{bmatrix} (2,2) \\ (0,0) \end{bmatrix}
		\ar[shift left]{r}{\Pop^1_s}
	\&
	\begin{bmatrix} (1,1) \\ (1,1) \end{bmatrix} = V'_1
		\ar[shift left]{l}{\Pop^s_1}
\end{tikzcd} ,
\end{gather}
where
\begin{align}
\Pop^{s}_1 \defeq{}&\begin{bmatrix} 
e^{1}{} \overset{0,0}{\odot}\bullet & e^{2}{} \overset{0,0}{\odot}\bullet\\
e^{1}{} \overset{1,1}{\odot}\bullet & e^{2}{} \overset{1,1}{\odot}\bullet
\end{bmatrix},
&
\Pop^1_{s} \defeq{}& 
\begin{bmatrix} 
\id_0 - \tfrac{1}{2}\Pop^{\parallel}_0 \Pop^0_{\parallel} & 0\\
0 & \id_0 - \tfrac{1}{2}\Pop^{\parallel}_0 \Pop^0_{\parallel} 
\end{bmatrix}
\begin{bmatrix} 
2e_{1}{} \overset{1,1}{\odot}\bullet & \tfrac{1}{2}e_{1}{} \overset{0,0}{\odot}\bullet\\
2e_{2}{} \overset{1,1}{\odot}\bullet & \tfrac{1}{2}e_{2}{} \overset{0,0}{\odot}\bullet
\end{bmatrix}.
\end{align}
We find
\begin{align} \label{eq:V1decomp}
\Pop^1_s \Pop^\top_1 =
\Pop^{s}_1 \Pop^a_{1'} =
\Pop^1_{\top} \Pop^{s}_1 =
\Pop^{1'}_{a} \Pop^1_{s} = 0,
&&
\Pop^{s}_1 \Pop^1_{s} + \Pop^{\top}_1 \Pop^1_{\top} = \id_1,
&&
\Pop^1_{s} \Pop^{s}_1 +  \Pop^a_{1'} \Pop^{1'}_a = \id_{1'}.
\end{align}
We also have the identity
\begin{align}
\Pop^{s}_1 \Pop^1_{s} \Pop^{s}_1 ={}&\Pop^{s}_1 .
\end{align}

To present the operators of the complex in a compact form, define the intermediate operators
\begin{align}
L_0 \defeq& \begin{bmatrix}
l_0 & \overline{l_0}
\end{bmatrix}, \qquad
\text{ with } 
l_0 \defeq e^{3}{} \overset{0,0}{\odot}\frac{2 \kappa_{1}{}^2 \bar{\kappa}_{1'}{}}{3 \Psi_{2}}\mathcal{K}^2 \mathcal{K}^2,\nonumber\\
L_1 \defeq& \begin{bmatrix}
l_1 & \overline{l_1}
\end{bmatrix}, \qquad
\text{ with } 
l_1 \defeq e_{1}{} \overset{1,1}{\odot}e_{2}{} \overset{1,0}{\odot}\frac{2}{3 \Psi_{2}}\mathcal{K}^1 \mathcal{K}^2, \nonumber\\
\linWeyl \defeq& \begin{bmatrix} 
\vartheta\Psi \\
\overline{\vartheta\Psi}
\end{bmatrix}, \qquad 
\text{ with }
\vartheta\Psi \defeq \begin{bmatrix}
\tfrac{1}{2}\sCurl \sCurl & - \tfrac{3}{32}\Psi_{2}\mathcal{K}^0 \mathcal{K}^0
\end{bmatrix},\nonumber\\
\linWeyl^A \defeq& \begin{bmatrix} 
\vartheta\Psi^{A} \\
\overline{\vartheta\Psi^{A}}
\end{bmatrix}, \qquad 
\text{ with } 
\vartheta\Psi^{A} \defeq \begin{bmatrix} 
- \tfrac{3}{4} \Psi_{2}\mathcal{K}^0 \mathcal{K}^1 e^{1}{} \overset{0,1}{\odot}\bullet & - \tfrac{3}{4} \Psi_{2}\mathcal{K}^0 \mathcal{K}^1 e^{2}{} \overset{0,1}{\odot}\bullet
\end{bmatrix},\nonumber\\
\linWeyl^D \defeq& \begin{bmatrix} 
\vartheta\Psi^{D} \\
\overline{\vartheta\Psi^{D}}
\end{bmatrix}, \qquad 
\text{ with }  \nonumber\\
\vartheta\Psi^{D} \defeq& \begin{bmatrix} 
0 & p^1 & 0 & p^2
\end{bmatrix}, \quad
 p^i \defeq - \tfrac{1}{2}e^{i}{}\overset{0,1}{\odot}\sTwist \bullet - \tfrac{1}{4}e^{1}{}\overset{0,1}{\odot}\FlatCon^{i}{}_{1}{}\overset{0,0}{\odot}\bullet - \tfrac{1}{4}e^{2}{}\overset{0,1}{\odot}\FlatCon^{i}{}_{2}{}\overset{0,0}{\odot}\bullet,
 \label{eq:IntermOps}
\end{align}
defined on the following spaces.
\begin{align}
L_0: W \to V_0, &&
L_1: W \to V_a, &&
\linWeyl: V_1 \to W, &&
\linWeyl^A: V'_1 \to W, &&
\linWeyl^D: V'_2 \to W,
\end{align}
where
\begin{align}
W\defeq \begin{bmatrix} (4,0) \\ (0,4) \end{bmatrix}
\end{align}
is equivalent to the space of 4-tensors with Weyl symmetries. 
\begin{lemma} \label{lem:OpIds}
The operators defined in \eqref{eq:IntermOps} satisfy the following identities
\begin{subequations} 
\begin{align}
L_0 \linWeyl \Kop_0 ={}& \Pop^\perp_0 \Pop^0_\perp, \label{eq:L0WeylK0id}\\
\linWeyl \Pop^s_1 ={}& \linWeyl^A + \linWeyl^D  \Kop'_1, \label{eq:WeylPs1id}\\
L_0 \linWeyl^A ={}& 0, \label{eq:L0WeylAid}\\
\Pop^1_{s} \Pop^{s}_1 -  \Pop^a_{1'} L_1 \mathcal{W}^A ={}&\id_{1'}, \label{eq:P1sPs1id}\\
\Pop^0_\parallel L_0 =&{} 0 \label{eq:P0parL0id}
\end{align}
\end{subequations} 
\end{lemma} 
\begin{proof}
The operator $L_0 \linWeyl \Kop_0$ applied to a vector yields the gauge dependence of $\vartheta\Psi_2$ multiplied by $e^{3}{} \overset{0,0}{\odot}\frac{2 \kappa_{1}{}^2 \bar{\kappa}_{1'}{}}{3 \Psi_{2}}$ and its complex conjugate. The gauge dependence is given by the 3 and 4 components of the vector and hence given by the right hand side of \eqref{eq:L0WeylK0id}. For \eqref{eq:WeylPs1id}, the projector $\Pop^s_1$ is commuted through the linearized Weyl operator $\linWeyl$. Part of it factors through $\Kop'_1$ and $\linWeyl^D$ is defined as the operator acting on it. The algebraic remainder is collected in $\linWeyl^A$. On $\SymSpin_{2,0}$ we have $ \mathcal{K}^2 \mathcal{K}^2 \mathcal{K}^0 = 0$. This together with the
complex conjugate version on $\SymSpin_{0,2}$ gives \eqref{eq:L0WeylAid}. Commuting $\mathcal{K}$ operators shows $L_1 \mathcal{W}^A = -\Pop^{1'}_a$ with $\Pop^{1'}_a$ given in \eqref{eq:Pop1paDef}. Then \eqref{eq:P1sPs1id} is the decomposition of $\id_{1'}$, given in \eqref{eq:V1decomp}. \eqref{eq:P0parL0id} follows directly from \eqref{eq:V0decomp}.
\end{proof}

Now we are prepared to define the remaining operators for the complex \eqref{diag:DiffComplex4D}. 
\begin{definition}
The operators in the first square of \eqref{diag:DiffComplex4D} are defined by
\begin{subequations} \label{eq:ComplexOps}
\begin{align}\label{eq:0ops}
\Cop_0 \defeq{}&\Pop^0_{\parallel}, &
\Dop_0 \defeq{}&\Pop^{\parallel}_0, &
\Kop_0 \defeq&\begin{bmatrix}
\sTwist \\
\sDiv 
\end{bmatrix}, &
\Hop_0 \defeq{}&  L_0 \linWeyl , &
\Hop'_0 \defeq{}&0.
\end{align}
The operators in the second square are defined by
\begin{align} \label{eq:1aops}
\Cop_1 \defeq{}& \left( \Pop^1_{s} -  \Pop^a_{1'} L_1 \linWeyl \right) \left( \id_1 -  \Kop_0 \Hop_0 \right), &
\Dop_1 \defeq{}& \Pop^{s}_1, &
\Kop_1 \defeq&
\begin{bmatrix}
\Kop'_1 \Cop_1 \\
\Pop^1_\top (\id_1 -  \Kop_0 \Hop_0)
 \end{bmatrix}, 
\end{align} 
\begin{align} \label{eq:1bops}
\Hop'_1 \defeq{}& \left( \Pop^1_{s} -  \Pop^a_{1'} L_1 \linWeyl \right) \Kop_0 L_0 \linWeyl^{D}  + \Pop^a_{1'} L_1 \linWeyl^{D}, &
\Hop_1 \defeq&\begin{bmatrix}
  0 & \Pop^{\top}_1
 \end{bmatrix}. &&
\end{align}
The operators in the third square are defined by
\begin{align}  \label{eq:2aops}
\Cop_2 \defeq&\begin{bmatrix}
\id_{2'} & 0
\end{bmatrix}, &
\Dop_2 \defeq& \begin{bmatrix}
\id_{2'} -  \Kop'_1 \Hop'_1 \\
- \Pop^1_\top \Kop_0 L_0 \linWeyl^{D}
\end{bmatrix}, &
\Kop_2 \defeq& \begin{bmatrix}
\Kop'_2 \Cop_2 \\
\Pop^{1'}_a  (\Hop'_1 \Cop_2 - \Cop_1 \Hop_1)\\
\Pop^0_\perp (\Hop_0 \Hop_1 +  L_0 \linWeyl^{D} \Cop_2)
\end{bmatrix},
\\
\Hop'_2 \defeq{}&0, & 
\Hop_2 \defeq& \begin{bmatrix}
0 &  \Kop'_1 \Pop^a_{1'}  & \Kop'_1 \Pop^1_s \Kop_0 \Pop^\perp_0 \\
0 & 0 & \Pop^1_\top \Kop_0 \Pop^\perp_0
\end{bmatrix}. &&\label{eq:2bops}
\end{align}
The operators in the fourth square are defined by
\begin{align} \label{eq:3ops}
\Cop_3 \defeq&\begin{bmatrix} 
\id_{3'} & 0 & 0
\end{bmatrix}, &
\Dop_3 \defeq&\begin{bmatrix} 
\id_{3'} \\
0 \\
0
\end{bmatrix}, &
\Kop_3 \defeq{}&\Kop'_3 \Cop_3 , & 
\Hop'_3 \defeq{}&0, &
\Hop_3 \defeq{}&0.
\end{align}
The operator between $V_4$ and $V'_4$ are defined by
\begin{align} \label{eq:4ops}
\Cop_4 \defeq{}& \id_{4'}, &
\Dop_4 \defeq{}& \id_{4'}.
\end{align}
\end{subequations}
\end{definition}

We are now ready to state and prove our main result.

\begin{theorem} \label{thm:kerr-compat}
The differential operators $\Kop_l$, $l=0,1,2,3$, defined in \eqref{eq:ComplexOps} constitute a compatibility complex for the Killing operator $\Kop = \Kop_0$ on the Kerr spacetime. The $\Kop_l$ complex is also locally exact.
\end{theorem}

\begin{proof}
The operators
defined in \eqref{eq:ComplexOps} constitute an equivalence up to homotopy of the complex $\Kop_l$ with the auxiliary complex $\Kop'_l$, which is the twisted de~Rham complex \eqref{eq:Kpops}, c.f. Definition~\ref{def:homalg} and \ref{def:flat-conn}. All the relevant identities follow from repeated application of projection identities and Lemma~\ref{lem:OpIds} as shown below. Due to Proposition~\ref{prp:de-rham}, $\Kop'_l$ is known both to be a compatibility complex and to be locally exact. Hence both of these properties also hold for the complex $\Kop_l$ by Proposition~\ref{prp:compat-sufficient}. The compositions of $\Kop'$ operators yield
\begin{align} \label{eq:KpIds}
\Kop'_1 \Kop'_0 = 0, &&
\Kop'_2 \Kop'_1 = 0, &&
\Kop'_3 \Kop'_2 = 0,
\end{align}
due to commutators and \eqref{eq:FlatConnectionConditions}. Next we present explicit derivations of all  required operator identities.
\renewcommand{\labelenumi}{\theenumi.}
\begin{enumerate}
\item \label{enu:id1} $\id_0 = \Dop_0 \Cop_0 + \Hop_0 \Kop_0$: 
\begin{align*}
\Dop_0 \Cop_0 + \Hop_0 \Kop_0 ={}& \Pop^{\parallel}_0 \Pop^0_{\parallel} + L_0 \linWeyl \Kop_0
=
 \Pop^{\parallel}_0 \Pop^0_{\parallel}  + \Pop^{\perp}_0 \Pop^0_{\perp}
= \id_0  \tag*{\text{ by (\ref{eq:0ops},\ref{eq:L0WeylK0id},\ref{eq:V0decomp})}}
\end{align*}
\item \label{enu:id2} $\id_{0'} = \Cop_0 \Dop_0 + \Hop'_0 \Kop'_0$:
\begin{align*}
\Cop_0 \Dop_0 + \Hop'_0 \Kop'_0 ={}&  \Pop^0_{\parallel} \Pop^{\parallel}_0 
=
\id_{0'} \tag*{ \text{ by (\ref{eq:0ops},\ref{eq:V0decomp})}}
\end{align*} 
\item \label{enu:id3} $\Kop_0 \Dop_0 = \Dop_1 \Kop'_0$:
\begin{align*}
\Kop_0 \Dop_0 ={}& \begin{bmatrix}
\sTwist ( e^{1}{} \overset{0,0}{\odot}\bullet ) &
\sTwist ( e^{2}{} \overset{0,0}{\odot}\bullet ) \\
\sDiv ( e^{1}{} \overset{0,0}{\odot}\bullet ) &
\sDiv ( e^{2}{} \overset{0,0}{\odot}\bullet )
\end{bmatrix} \\
={}& \begin{bmatrix}
(\sTwist e^{1}) \overset{0,0}{\odot}\bullet + e^{1}  \overset{0,0}{\odot}(\sTwist \bullet) &
(\sTwist e^{2} ) \overset{0,0}{\odot}\bullet  + e^{2}{} \overset{0,0}{\odot}(\sTwist \bullet )\\
(\sDiv  e^{1}) \overset{0,0}{\odot}\bullet  + e^{1} \overset{1,1}{\odot}(\sTwist \bullet ) &
(\sDiv e^{2} ) \overset{0,0}{\odot}\bullet  + e^{2}  \overset{1,1}{\odot}(\sTwist \bullet )
\end{bmatrix}
\\
={}& \Dop_1 \Kop'_0
\end{align*} 
In the last step we used $\sTwist e^i = -\sum_{l=1}^2  e^l \overset{0,0}{\odot} \FlatCon^i{}_l$, $\sDiv e^i = 0$ and $e^i \overset{1,1}{\odot} \FlatCon^j{}_k = 0$ for $i,j,k = 1,2$ and the definition \eqref{eq:Kpops}.
\item \label{enu:id4} $\Cop_1 \Kop_0 = \Kop'_0 \Cop_0$:
\begin{align*}
\Cop_1 \Kop_0 ={}& \left( \mathsf{P}^1_{s} -  \Pop^a_{1'} L_1 \linWeyl \right) \left( \id_1 -  \Kop_0 \Hop_0 \right)\Kop_0 
 \tag*{ \text{ by \eqref{eq:1aops}}} \\
={}& \left( \mathsf{P}^1_{s} -  \Pop^a_{1'} L_1 \linWeyl \right) \Kop_0 \left( \id_0 - \Hop_0 \Kop_0 \right) &\\
={}& \left( \mathsf{P}^1_{s} -  \Pop^a_{1'} L_1 \linWeyl \right) \Kop_0 \Dop_0 \Cop_0 
 \tag*{ \text{ by \ref{enu:id1}.}} \\
={}& \left( \mathsf{P}^1_{s} -  \Pop^a_{1'} L_1 \linWeyl \right) \Pop^s_1 \Kop'_0 \Cop_0 
 \tag*{ \text{ by \ref{enu:id3}., \eqref{eq:1aops}}} \\
={}& \left( \mathsf{P}^1_{s} \Pop^s_1 -  \Pop^a_{1'} L_1 \linWeyl^A -  \Pop^a_{1'} L_1 \linWeyl^D  \Kop'_1 \right) \Kop'_0 \Cop_0 
\tag*{ \text{ by \eqref{eq:WeylPs1id}}} \\
={}&  \Kop'_0 \Cop_0
\tag*{ \text{ by \eqref{eq:P1sPs1id}, \eqref{eq:KpIds}}}
\end{align*} 
\item \label{enu:id5} $\id_{1'} = \Cop_1 \Dop_1 + \Hop'_1 \Kop'_1 + \Kop'_0 \Hop'_0$:
\begin{align*}
\Cop_1 \Dop_1 + \Hop'_1 \Kop'_1 + \Kop'_0 \Hop'_0 ={}&  \left( \mathsf{P}^1_{s} -  \Pop^a_{1'} L_1 \linWeyl \right) \left( \id_1 -  \Kop_0 L_0 \linWeyl \right) \Pop^s_1  &\\
&+ \left( \mathsf{P}^1_{s} -  \Pop^a_{1'} L_1 \linWeyl \right) \Kop_0 L_0 \linWeyl^{D} \Kop'_1 + \Pop^a_{1'} L_1 \linWeyl^{D}\Kop'_1 
\tag*{ \text{ by \eqref{eq:ComplexOps}}} \\
={}&  \mathsf{P}^1_{s}  \mathsf{P}^s_{1} -  \Pop^a_{1'} L_1 \linWeyl^A 
\tag*{ \text{ by \eqref{eq:WeylPs1id},\eqref{eq:L0WeylAid}}} \\
={}& \id_{1'}
\tag*{ \text{ by \eqref{eq:P1sPs1id}}}
\end{align*}
\item \label{enu:id6} $\id_1 = \Dop_1 \Cop_1 + \Hop_1 \Kop_1 + \Kop_0 \Hop_0$:
\begin{align*}
\Dop_1 \Cop_1 + \Hop_1 \Kop_1 + \Kop_0 \Hop_0 ={}& \Pop^s_1  \left( \mathsf{P}^1_{s} -  \Pop^a_{1'} L_1 \linWeyl \right) \left( \id_1 -  \Kop_0 \Hop_0 \right) & \\
& + \Pop^\top_1 \Pop^1_\top \left( \id_1 - \Kop_0 \Hop_0 \right) + \Kop_0 \Hop_0
\tag*{ \text{ by \eqref{eq:1aops},\eqref{eq:1bops}}} \\
={}& \id_1 
\tag*{ \text{ by \eqref{eq:V1decomp}}} 
\end{align*} 
\item \label{enu:id7} $\Kop'_1 \Cop_1 = \Cop_2 \Kop_1$ \hfill  by \eqref{eq:1aops},\eqref{eq:2aops}
\item \label{enu:id8} $\Kop_1 \Dop_1 = \Dop_2 \Kop'_1$:
\begin{align*}
 \Dop_2 \Kop'_1 ={}& \begin{bmatrix}
 \Kop'_1(\id_{1'} -   \Hop'_1 \Kop'_1) \\
 - \Pop^1_\top \Kop_0 L_0 \linWeyl^{D}\Kop'_1
 \end{bmatrix}
 \tag*{ \text{ by \eqref{eq:2aops}}} \\
 ={}& \begin{bmatrix}
  \Kop'_1 \Cop_1 \Dop_1 \\
  - \Pop^1_\top \Kop_0 L_0 \linWeyl \Dop_1
  \end{bmatrix}
  \tag*{ \text{ by \ref{enu:id5}., \eqref{eq:WeylPs1id}, \eqref{eq:L0WeylAid}, \eqref{eq:1aops}}} \\
  ={}& \Kop_1 \Dop_1 
  \tag*{ \text{ by \eqref{eq:1aops}, \eqref{eq:V1decomp}, \eqref{eq:0ops}}}
\end{align*} 
\item \label{enu:id9} $\id_{2'} = \Cop_2 \Dop_2 + \Hop'_2 \Kop'_2 + \Kop'_1 \Hop'_1$ \hfill by \eqref{eq:2aops},\eqref{eq:2bops}
\item \label{enu:id10} $\id_2 = \Dop_2 \Cop_2 + \Hop_2 \Kop_2 + \Kop_1 \Hop_1$:
\begin{align*}
\Dop_2 \Cop_2 + \Kop_1 \Hop_1 + \Hop_2 \Kop_2  ={}&  
\begin{bmatrix}
(\id_{2'} -  \Kop'_1 \Hop'_1) \Cop_2 \\
- \Pop^1_\top \Kop_0 L_0 \linWeyl^{D} \Cop_2
\end{bmatrix} +
\begin{bmatrix}
\Kop'_1 \Cop_1 \Hop_1 \\
\Pop^1_\top(\id_1 - \Kop_0 \Hop_0) \Hop_1
\end{bmatrix}\\
& + 
\begin{bmatrix}
\Kop'_1 \Pop^a_{1'} \Pop^{1'}_a (\Hop'_1 \Cop_2 - \Cop_1 \Hop_1) + \Kop'_1 \Pop^1_s \Kop_0 \Pop^\perp_0 \Pop^0_\perp (\Hop_0 \Hop_1 + L_0 \linWeyl^D \Cop_2) \\
\Pop^1_\top \Kop_0 \Pop^\perp_0 \Pop^0_\perp (\Hop_0 \Hop_1 + L_0 \linWeyl^D \Cop_2)
\end{bmatrix} \\
={}& 
\begin{bmatrix}
\Kop'_1 (\Pop^a_{1'} \Pop^{1'}_a - \id_{1'}) (\Hop'_1 \Cop_2 - \Cop_1 \Hop_1) + \Kop'_1 \Pop^1_s \Kop_0 \Pop^\perp_0 \Pop^0_\perp (\Hop_0 \Hop_1 + L_0 \linWeyl^D \Cop_2) + \Cop_2\\
\Pop^1_\top \Kop_0 (\Pop^\perp_0 \Pop^0_\perp - \id_0) (\Hop_0 \Hop_1 + L_0 \linWeyl^D \Cop_2) + \Pop^1_\top \Hop_1
\end{bmatrix} \\
={}& 
\begin{bmatrix}
\Kop'_1 \Pop^1_s ( \Kop_0 L_0 \linWeyl^D \Cop_2 + \Hop_1 - \Kop_0 \Hop_0 \Hop_1 ) + \Kop'_1 \Pop^1_s \Kop_0 (\Hop_0 \Hop_1 + L_0 \linWeyl^D \Cop_2) + \Cop_2\\
 \Pop^1_\top \Hop_1
\end{bmatrix} \\
={}&
\begin{bmatrix}
- \Kop'_1 \Pop^1_s \Hop_1  + \Cop_2\\
 \Pop^1_\top \Hop_1
\end{bmatrix} =
\id_2
\end{align*}
Here we used \eqref{eq:ComplexOps}, \eqref{eq:V0decomp}, \eqref{eq:P0parL0id}, \eqref{eq:V1decomp}.
\item \label{enu:id11} $\Kop_2 \Dop_2 = \Dop_3 \Kop'_2$:
\begin{align*}
\Kop_2 \Dop_2 ={}& \begin{bmatrix}
\Kop'_2 \Cop_2 \Dop_2\\
\Pop^{1'}_a  (\Hop'_1 \Cop_2 \Dop_2 - \Cop_1 \Hop_1 \Dop_2)\\
\Pop^0_\perp (\Hop_0 \Hop_1 \Dop_2 +  L_0 \linWeyl^{D} \Cop_2 \Dop_2)
\end{bmatrix}
\tag*{ \text{ by \eqref{eq:2aops}}} \\
={}& \begin{bmatrix}
\Kop'_2 (\id_{2'} - \Kop'_1 \Hop'_1)\\
\Pop^{1'}_a  (\Hop'_1 (\id_{2'} - \Kop'_1 \Hop'_1) - \Cop_1 \Hop_1 \Dop_2)\\
\Pop^0_\perp (\Hop_0 \Hop_1 \Dop_2 +  L_0 \linWeyl^{D} (\id_{2'} - \Kop'_1 \Hop'_1))
\end{bmatrix}
\tag*{ \text{ by \ref{enu:id9}.}} \\
={}& \begin{bmatrix}
\Kop'_2\\
\Pop^{1'}_a  ((\id_{1'} - \Hop'_1 \Kop'_1) \Hop'_1 + \Cop_1 \Pop^\top_1 \Pop^1_\top \Kop_0 L_0 \linWeyl^D)\\
\Pop^0_\perp (-\Hop_0 \Pop^\top_1 \Pop^1_\top \Kop_0 L_0 \linWeyl^D +  L_0 \linWeyl^{D}  - L_0 \linWeyl^{D} \Kop'_1 \Hop'_1)
\end{bmatrix}
\tag*{ \text{ by \eqref{eq:KpIds}, \eqref{eq:1bops}, \eqref{eq:2aops}}} \\
={}& \begin{bmatrix}
\Kop'_2\\
\Pop^{1'}_a  ( \Cop_1 \Dop_1 \Hop'_1 + \Cop_1 \Pop^\top_1 \Pop^1_\top \Kop_0 L_0 \linWeyl^D)\\
\Pop^0_\perp (-\Hop_0 \Pop^\top_1 \Pop^1_\top \Kop_0 L_0 \linWeyl^D +  L_0 \linWeyl^{D}  - L_0 \linWeyl \Pop^s_1 \Hop'_1)
\end{bmatrix}
\tag*{ \text{ by \ref{enu:id5}., \eqref{eq:WeylPs1id}, \eqref{eq:L0WeylAid}}} \\
={}& \begin{bmatrix}
\Kop'_2\\
\Pop^{1'}_a \Cop_1 ( \Pop^s_1 \Pop^1_s  + \Pop^\top_1 \Pop^1_\top )\Kop_0 L_0 \linWeyl^D\\
\Pop^0_\perp (-\Hop_0 \Pop^\top_1 \Pop^1_\top \Kop_0 L_0 \linWeyl^D +  L_0 \linWeyl^{D}  - L_0 \linWeyl \Pop^s_1 \Pop^1_s \Kop_0 L_0 \linWeyl^D)
\end{bmatrix}
\tag*{ \text{ by \eqref{eq:1bops}, \eqref{eq:V1decomp}}} \\
={}& \begin{bmatrix}
\Kop'_2\\
\Pop^{1'}_a \Cop_1 \Kop_0 L_0 \linWeyl^D\\
\Pop^0_\perp (-\Hop_0 \Kop_0 L_0 \linWeyl^D +  L_0 \linWeyl^{D})
\end{bmatrix}
\tag*{ \text{ by \eqref{eq:V1decomp}, \eqref{eq:0ops}}} \\
={}& \begin{bmatrix}
\Kop'_2\\
\Pop^{1'}_a \Kop'_0 \Cop_0 L_0 \linWeyl^D\\
\Pop^0_\perp \Dop_0 \Cop_0 L_0 \linWeyl^D 
\end{bmatrix} =
\begin{bmatrix}
\Kop'_2\\
0\\
0 
\end{bmatrix} =
\Dop_3 \Kop'_2
\tag*{ \text{ by \ref{enu:id1}., \ref{enu:id4}., \eqref{eq:P0parL0id}, \eqref{eq:0ops}}}
\end{align*}
\item \label{enu:id12} $\Kop'_2 \Cop_2 = \Cop_3 \Kop_2$ \hfill by \eqref{eq:2aops}, \eqref{eq:3ops}
\item \label{enu:id13} $\id_{3'} = \Cop_3 \Dop_3 + \Hop'_3 \Kop'_3 + \Kop'_2 \Hop'_2$ \hfill by \eqref{eq:2bops}, \eqref{eq:3ops}
\item \label{enu:id14} $\id_3 = \Dop_3 \Cop_3 + \Hop_3 \Kop_3 + \Kop_2 \Hop_2$:
\begin{align*}
 \Dop_3 \Cop_3 + \Kop_2 \Hop_2 ={}& 
 \begin{bmatrix}
 \id_{3'} & 0 & 0 \\ 0 & 0 & 0 \\ 0 & 0 & 0
 \end{bmatrix} + \begin{bmatrix}
 \Kop'_2 \Cop_2 \Hop_2 \\
 \Pop^{1'}_a  (\Hop'_1 \Cop_2 \Hop_2 - \Cop_1 \Hop_1 \Hop_2)\\
 \Pop^0_\perp (\Hop_0 \Hop_1 \Hop_2 +  L_0 \linWeyl^{D} \Cop_2 \Hop_2)
 \end{bmatrix}
 \tag*{ \text{ by \eqref{eq:2aops}, \eqref{eq:3ops}}} \\
 ={}& \begin{bmatrix}
  \id_{3'} & \Kop'_2 \Kop'_1  \Pop^a_{1'} & \Kop'_2 \Kop'_1 \Pop^1_s \Kop_0 \Pop^{\perp}_0 \\ 
  0 &  \Pop^{1'}_a \Hop'_1 \Kop'_1  \Pop^a_{1'} & \Pop^{1'}_a (\Hop'_1 \Kop'_1 \Pop^1_s \Kop_0 \Pop^{\perp}_0  - \Cop_1 \Pop^\top_1 \Pop^1_\top \Kop_0 \Pop^\perp_0) \\
  0 &  \Pop^0_\perp L_0 \linWeyl^{D} \Kop'_1  \Pop^a_{1'} & \Pop^0_\perp (L_0 \linWeyl^{D} \Kop'_1 \Pop^1_s \Kop_0 \Pop^{\perp}_0 + \Hop_0 \Pop^\top_1 \Pop^1_\top \Kop_0 \Pop^\perp_0)\\
  \end{bmatrix}
   \tag*{ \text{ by \eqref{eq:ComplexOps}}} \\
 ={}& \begin{bmatrix}
   \id_{3'} & 0 & 0 \\ 
   0 &  \Pop^{1'}_a (\id_{1'} - \Cop_1 \Dop_1)  \Pop^a_{1'} & \Pop^{1'}_a ( (\id_{1'} - \Cop_1 \Dop_1) \Pop^1_s   - \Cop_1 \Pop^\top_1 \Pop^1_\top ) \Kop_0 \Pop^\perp_0 \\
   0 &  \Pop^0_\perp L_0 \linWeyl \Pop^s_1  \Pop^a_{1'} & \Pop^0_\perp (L_0 \linWeyl \Pop^s_1 \Pop^1_s + \Hop_0 \Pop^\top_1 \Pop^1_\top ) \Kop_0 \Pop^\perp_0\\
   \end{bmatrix}  
    \tag*{ \text{ by \eqref{eq:KpIds}, \ref{enu:id5}, \eqref{eq:WeylPs1id}}} \\
  ={}& \begin{bmatrix}
    \id_{3'} & 0 & 0 \\ 
    0 & \id_a & - \Pop^{1'}_a \Cop_1 \Kop_0 \Pop^\perp_0 \\
    0 &  0 & \Pop^0_\perp \Hop_0 \Kop_0 \Pop^\perp_0\\
    \end{bmatrix} =
 \begin{bmatrix}
     \id_{3'} & 0 & 0 \\ 
     0 & \id_a & - \Pop^{1'}_a \Kop'_0 \Cop_0 \Pop^\perp_0 \\
     0 &  0 & \Pop^0_\perp (\id_0 -\Dop_0 \Cop_0) \Pop^\perp_0\\
     \end{bmatrix}  = \id_3   
     \tag*{ \text{ by \eqref{eq:ComplexOps}, \eqref{eq:V1decomp}, \eqref{eq:Ida}, \ref{enu:id4}., \ref{enu:id1}., \eqref{eq:V0decomp} }}
\end{align*}
\item \label{enu:id15} $\Kop_3 \Dop_3 = \Dop_4 \Kop'_3$ \hfill by \eqref{eq:3ops}, \eqref{eq:4ops}
\item \label{enu:id16} $\Kop'_3 \Cop_3 = \Cop_4 \Kop_3$ \hfill by \eqref{eq:3ops}, \eqref{eq:4ops}
\end{enumerate}

Compositions of $\Kop$ operators yield
\begin{subequations} 
\begin{align*} 
\Kop_1 \Kop_0 ={}& \begin{bmatrix}
\Kop'_1 \Cop_1 \Kop_0 \\
\Pop^1_\top (\id_1 - \Kop_0 \Hop_0) \Kop_0
\end{bmatrix} = 
 \begin{bmatrix}
\Kop'_1 \Kop'_0 \Cop_0 \\
\Pop^1_\top \Kop_0 \Dop_0 \Cop_0
\end{bmatrix}= 
 \begin{bmatrix}
0 \\
\Pop^1_\top \Pop^s_1 \Kop'_0 \Cop_0
\end{bmatrix} = 0, \\
\Kop_2 \Kop_1 ={}&  \begin{bmatrix}
\Kop'_2 \Cop_2 \Kop_1 \\
\Pop^{1'}_a  (\Hop'_1 \Cop_2 \Kop_1 - \Cop_1 \Hop_1 \Kop_1)\\
\Pop^0_\perp (\Hop_0 \Hop_1 \Kop_1 +  L_0 \linWeyl^{D} \Cop_2 \Kop_1)
\end{bmatrix} \nonumber \\
={}&
 \begin{bmatrix}
\Kop'_2 \Kop'_1 \Cop_1 \\
\Pop^{1'}_a  (\Hop'_1 \Kop'_1 \Cop_1 - \Cop_1 \Hop_1 \Kop_1)\\
\Pop^0_\perp (\Hop_0 \Hop_1 \Kop_1 +  L_0 \linWeyl^{D} \Kop'_1 \Cop_1)
\end{bmatrix} \nonumber \\
={}&
 \begin{bmatrix}
0 \\
\Pop^{1'}_a  \Cop_1( \id_1 - \Dop_1 \Cop_1 - \Hop_1 \Kop_1)\\
\Pop^0_\perp \Hop_0 (\Hop_1 \Kop_1 + \Dop_1 \Cop_1)
\end{bmatrix} \nonumber \\
={}&
 \begin{bmatrix}
0 \\
\Pop^{1'}_a  \Cop_1 \Kop_0 \Hop_0\\
-\Pop^0_\perp \Hop_0 (\Kop_0 \Hop_0 - \id_1)
\end{bmatrix}
=
 \begin{bmatrix}
0 \\
\Pop^{1'}_a \Kop'_0 \Cop_0 \Hop_0\\
\Pop^0_\perp \Dop_0 \Cop_0 \Hop_0 
\end{bmatrix} = 0, \\
\Kop_3 \Kop_2 ={}& \Kop'_3 \Cop_3 \Kop_2 =  \Kop'_3 \Kop'_2 \Cop_2 = 0. 
\end{align*}
\end{subequations} 
\end{proof}

\section{Equivalence of invariants} \label{sec:equiv}

Here we discuss the equivalence of $\Kop_1$ and the operators appearing in the set of invariants of \cite{ab-kerr} which we denote collectively by $\widetilde \Kop_1$. Recall that $\widetilde \Kop_1$ consists of the quantities\footnote{There is a typo in the GHP form of $\mathbb{I}_\zeta$ given in equation (26) of \cite{ab-kerr}. In two instances the factor $p_+ p_-$ should be replaced by $(p^2 + \bar{p}^2)$.}
\begin{subequations} \label{eq:PRLinv}
\begin{align} 
\vartheta \Phi &: V_1 \to (2,2), \\
\vartheta \Lambda &: V_1 \to (0,0), \\
\mathcal{P}^{2} \vartheta \Psi &: V_1 \to (4,0), \\
\mathbb{I}_V &: V_1 \to (0,0), \qquad  \text{ for } V \in \{\xi,\zeta\}.
\end{align}
\end{subequations} 
The explicit forms of the operators are given in \eqref{eq:LinCurvOps} and by
\begin{subequations} 
\begin{align}
\mathcal{P}^{2} \vartheta \Psi \defeq{}& \begin{bmatrix} 
\tfrac{1}{2}\mathcal{P}^{2} \sCurl \sCurl
&
0
\end{bmatrix}, \\
\InvSymb_V
\defeq{}&
\InvSymb_V^{D}\linWeyl +  \InvSymb_V^{A}, \label{eq:IIIvDef1}
\end{align}
\end{subequations} 
with $\mathcal{P}^2$ given in \eqref{eq:P2Def} and
\begin{subequations}
\begin{align}
\InvSymb_V^{A}\defeq{}& \tfrac{81}{2} \Psi_{2} \kappa_{1}{}^3\begin{bmatrix}
( V \overset{1,1}{\odot} \xi \overset{1,1}{\odot} \mathcal{K}^0 \mathcal{K}^2\bullet  
-  V  \overset{1,1}{\odot} \xi  \overset{1,1}{\odot} \bullet  )
&
\tfrac{1}{4}  V \overset{1,1}{\odot} \xi \overset{0,0}{\odot} \bullet
\end{bmatrix},\\
\InvSymb_V^{D}\defeq{}&\begin{bmatrix}
\InvSymb_V^{D1}
&\InvSymb_V^{D2}
\end{bmatrix},\\
\InvSymb_V^{D1}\defeq{}&-81V {\overset{1,1}{\odot}}\mathcal{K}^1\sTwist( \kappa_{1}{}^4\mathcal{K}^2\mathcal{K}^2\bullet) - \tfrac{81}{2}\kappa_{1}{}^4(\mathcal{K}^2\sCurl V)_{}\mathcal{K}^2\mathcal{K}^2\bullet - 27\kappa_{1}{}^3V {\overset{1,1}{\odot}}\xi {\overset{1,0}{\odot}}\mathcal{K}^1\mathcal{K}^2\bullet ,\\
\InvSymb_V^{D2}\defeq{}& - \tfrac{81}{2}\bar{\kappa}_{1'}{}^4(\overline{\mathcal{K}}^2\sCurlDagger V)_{}\overline{\mathcal{K}}^2\overline{\mathcal{K}}^2\bullet + 27\bar{\kappa}_{1'}{}^3V {\overset{1,1}{\odot}}\xi {\overset{0,1}{\odot}}\overline{\mathcal{K}}^1\overline{\mathcal{K}}^2\bullet.
\end{align}
\end{subequations}
Applied to a linearized metric, $\vartheta \Phi$ and $\vartheta \Lambda$ are the trace-free and trace parts of linearized Ricci, $\mathcal{P}^{2} \vartheta \Psi$ are the linearized extreme Weyl components (Teukolsky scalars), cf. \eqref{eq:LinCurvOps}. $\mathbb{I}_V$ are two third order complex scalar invariants.

Neither the number of components nor the differential order of the two sets of invariants coincide and we refer to section~\ref{sec:counting} for more details. Next we show that $\widetilde \Kop_1$ and  $\Kop_1$ can be factored through each other, thereby proving the completeness of $\widetilde \Kop_1$.

\subsection{Factorization of \texorpdfstring{$\widetilde \Kop_1$}{tilde K1} through \texorpdfstring{$\Kop_1$}{K1}} \label{sec:C2Comps}
In this subsection we show that $\widetilde \Kop_1= \widetilde \Dop_2 \circ \Kop_1$ for some operator $\widetilde \Dop_2$. Following \cite[Lemma 4]{kh-compat}, assume we have $\widetilde \Kop_1 \circ \Dop_1 = \Fop \circ \Kop_1'$, for some operator $\Fop$, then
\begin{align} \label{eq:PRLThroughK1Factorization}
\widetilde \Kop_1
={}& \widetilde \Kop_1 \circ (\Dop_1 \circ  \Cop_1 + \Hop_1 \circ \Kop_1 + \Kop_0 \circ \Hop_0) \nonumber \\
={}& \Fop \circ \Kop_1' \circ \Cop_1 + \widetilde \Kop_1 \circ \Hop_1 \circ \Kop_1 \nonumber \\
={}& (\Fop \circ \Cop_2 + \widetilde \Kop_1 \circ \Hop_1) \circ \Kop_1
\end{align}
Hence, we can choose $\widetilde \Dop_2 = \Fop \circ \Cop_2 + \widetilde \Kop_1 \circ \Hop_1$. Now we compute $\Fop$. 

\begin{definition}
\begin{subequations} \label{eq:defPRLfactorization}
\begin{align}
\vartheta\Phi^{D} \defeq{}& \begin{bmatrix}
k^1 & \overline{k^1} & k^2 & \overline{k^2}
\end{bmatrix}, \text{ with }  \\
k^i \defeq{}& \tfrac{1}{6}e^{i}{}{\overset{0,0}{\odot}}\sCurl\bullet -  \tfrac{1}{4}e^{i}{}{\overset{0,1}{\odot}}\sTwist\bullet -  \tfrac{1}{8}e^{1}{}{\overset{0,1}{\odot}}\FlatCon^{i}{}_{1}{}{\overset{0,0}{\odot}}\bullet -  \tfrac{1}{8}e^{2}{}{\overset{0,1}{\odot}}\FlatCon^{i}{}_{2}{}{\overset{0,0}{\odot}}\bullet + \tfrac{1}{12}e^{1}{}{\overset{0,0}{\odot}}\FlatCon^{i}{}_{1}{}{\overset{0,1}{\odot}}\bullet + \tfrac{1}{12}e^{2}{}{\overset{0,0}{\odot}}\FlatCon^{i}{}_{2}{}{\overset{0,1}{\odot}} \bullet, \nonumber\\
 \vartheta\Lambda^{D} \defeq{}& \begin{bmatrix}
 r^1 & \overline{r^1} & r^2 & \overline{r^2}
 \end{bmatrix}, \text{ with } \\
  r^i \defeq{}& - \tfrac{1}{24}e^{i}{}{\overset{1,1}{\odot}}\sCurl\bullet -  \tfrac{1}{48}e^{1}{}{\overset{1,1}{\odot}}\FlatCon^{i}{}_{1}{}{\overset{0,1}{\odot}}\bullet -  \tfrac{1}{48}e^{2}{}{\overset{1,1}{\odot}}\FlatCon^{i}{}_{2}{}{\overset{0,1}{\odot}}\bullet. \nonumber
\end{align}
\end{subequations} 
\end{definition}

\begin{lemma}
The operator $\widetilde \Kop_1$ composed with $\Dop_1$ factors through $\Kop_1'$ via $\widetilde \Kop_1\circ\Dop_1=\Fop\circ\Kop_1'$, where
\begin{align}
\Fop={}&\begin{bmatrix}
\vartheta\Phi^{D}\\
\vartheta\Lambda^{D}\\
\mathcal{P}^{2} \vartheta \Psi^{D}\\
\InvSymb_\xi^{D} \linWeyl^{D}\\
\InvSymb_\zeta^{D} \linWeyl^{D}
\end{bmatrix}
\end{align}
\end{lemma}
\begin{proof}
The relations $\vartheta\Phi \Dop_1 = \vartheta\Phi^{D} \Kop_1'$ and $\vartheta\Lambda \Dop_1 = \vartheta\Lambda^{D} \Kop_1'$ are commutators of the linearized trace-free Ricci operator  and linearized Ricci scalar operator \eqref{eq:LinCurvOps} with the algebraic operator $\Dop_1$. For the Weyl components $\vartheta \Psi$ we use $\mathcal{P}^{2} \vartheta \Psi^A = 0$ so that the factorization $\mathcal{P}^{2} \vartheta \Psi \Dop_1 = \mathcal{P}^{2} \vartheta \Psi^{D} \Kop_1'$ follows from \eqref{eq:WeylPs1id}. 
For $\InvSymb_V$ we find 
\begin{align}
\InvSymb_V^{A} \Dop_1 ={}&\InvSymb_V^{D} \linWeyl^{A} = 0, \qquad  \text{ for } V \in \{\xi,\zeta\},
\end{align}
so that \eqref{eq:IIIvDef1} with \eqref{eq:WeylPs1id} leads to 
$\InvSymb_V \Dop_1= \InvSymb_V^{D} \linWeyl^{D} \Kop'_1$ for $V \in \{\xi,\zeta\}$.
\end{proof}

\subsection{Factorization of \texorpdfstring{$\Kop_1$}{K1} through \texorpdfstring{$\widetilde \Kop_1$}{tilde K1}} \label{sec:C2tildeComps}
In this subsection we show that $\Kop_1 = \widetilde \Cop_2 \circ \widetilde \Kop_1$ for some operator $\widetilde C_2 $. For the second component of $\Kop_1$ we find the relation
\begin{align}
\Pop^1_\top (\id_1 -  \Kop_0 \Hop_0) 
\begin{bmatrix}
G\\
\slashed{G}
\end{bmatrix} 
=
- \frac{2 \bar{\kappa}_{1'}{}^2}{729 \Psi_{2} \kappa_{1}{}}
\begin{bmatrix}
\InvSymb_\xi \\
   \frac{ (\InvSymb_\zeta + \overline{\mathbb{I}_{\zeta}{}} ) }{9\kappa_{1}{}  \bar{\kappa}_{1'} }  +  \frac{(\InvSymb_\xi +\overline{\mathbb{I}_{\xi}{}})(\bar{\kappa}_{1'}{}^2+\kappa_1^2)}{4 \kappa_{1}  \bar{\kappa}_{1'}}\\
 \overline{\mathbb{I}_{\xi}{}}
\end{bmatrix}.
\end{align}

The expansion of the first component of $\Kop_1$ in terms of $\widetilde{\Kop}_1$ required a long computation and the result is displayed in appendix~\ref{sec:AppC2tildeComps}.

\section{Counting invariants} \label{sec:counting} 

The conclusion from Section~\ref{sec:equiv} is that the components of both
operators $\Kop_1$ (computed using the methods of~\cite{kh-compat}) and
$\tilde{\Kop}_1$ in \eqref{eq:PRLinv} each constitute a
complete set of local first order gauge invariants for metric
perturbations of Kerr. Yet the two operators look quite different:
$\Kop_1$ is of 4th differential order%
	\footnote{ The linearized Weyl operator composed with the Killing operator, $\linWeyl \Kop_0$ reduces (in any geometry) from differential order 3 to 1 by using commutators. Therefore $\Kop_1$  as defined in \eqref{eq:1aops} reduces from order 6 to 4.} %
and has 15 real components, while
$\tilde{\Kop}_1$ is of 3rd differential order and has 18 real
components. So neither the degree nor the number of components is a
stable quantity for a complete set of invariants under the kind of
equivalence considered in Sections~\ref{sec:compat} and~\ref{sec:equiv}.
A natural question to ask is the following: is there any stable way to
assign either an order or a number of components to a complete set of
invariants, perhaps under some condition of minimality?

Practically speaking, the higher order of $\Kop_1$ is what allows it to
get away with fewer components than $\tilde{\Kop}_1$. Taking
differential linear combinations of the high order components of
$\Kop_1$ it is possible to cancel the highest order
coefficients, leaving behind the extra lower order components that are
present in $\tilde{\Kop}_1$ but not in $\Kop_1$. It is also easy to see
how, even without changing the number of components, the order of either
$\Kop_1$ or $\tilde{\Kop}_1$ could be artificially inflated by mixing a
high order derivative of one component with another, in an invertible
way. The ambiguity in the order and in the number of components
lies in the subtle interplay between the leading and sub-leading order
coefficients in the gauge invariants.

This issue is very well known in the literature on overdetermined PDEs,
where a set of relevant homological invariants has been identified,
so-called \emph{Spencer cohomologies} of a differential
operator~\cite{spencer, goldschmidt-lin, seiler-inv}. These invariants
are stable under the kind of equivalence considered in
Sections~\ref{sec:compat} and~\ref{sec:equiv}, and the dimensions of certain Spencer cohomologies can
be combined to give the order and the number of components of the
differential operator, provided it is in so-called \emph{involutive} and
\emph{minimal} form. Minimality is a simple condition on the principal
symbol, while involutivity is a more sophisticated condition involving
both the principal and sub-principal parts of the operator.

The classical principal symbol $\sigma_p(\Kop)$ of an operator $\Kop$
(the coefficients of highest differential order, contracted with
covectors $p\in T^*M$ instead of partial derivatives) is most useful
when every component of $\Kop$ has the same order. For operators of
mixed orders, it is more useful to work with the \emph{graded
symbol}~\cite{kl-review, seiler-inv}, which is essentially the same as
the \emph{weighted symbol} of~\cite{douglis-nirenberg}. Using the same
notation, the graded symbol $\sigma_p(\Kop)$ is a matrix of homogeneous
polynomials in the covector $p\in T^*M$, possibly of different orders
but obeying some simple bounds, and it satisfies the convenient identity
\begin{equation} \label{eq:gradsymb-comp}
	\sigma_p(\Lop \Kop) = \sigma_p(\Lop) \sigma_p(\Kop) ,
\end{equation}
even if the right-hand side is zero. The operator $\Kop$ is
\emph{minimal}%
	\footnote{This usage is compatible with the notion of a \emph{minimal
	resolution} in commutative algebra~\cite{seiler-inv}.} %
when the rows of its symbol $\sigma_p(\Kop)$ are linearly
independent with respect to $p$-independent coefficients. And $\Kop$ is
\emph{involutive}%
	\footnote{The traditional definition of the notion of \emph{involutivity} of a
	differential operator or a PDE is rather technical. Ours is simplified
	and synthesized from the relation between traditional involutivity and
	commutative algebra elaborated in~\cite{kl-review, seiler-inv}.}
when any admissible matrix of homogeneous polynomials
$\lambda_p$ satisfying $\lambda_p \sigma_p(\Kop)=0$ can be extended to a
differential operator $\Lop$ (that is, $\sigma_p(\Lop) = \lambda_p$)
such that $\Lop\Kop = 0$.

Minimality is easy to check, as it follows from maximality of numerical
rank when $\sigma_p(\Kop)$ is evaluated on a generic value of $p$. On
the other hand, in general, it is much easier to show that $\Kop$ is not
involutive (meaning that its degree or number of components has no
invariant meaning) than the opposite. As an illustration, let us consider the
involutivity of the Killing operator $\Kop_0$ on Kerr.
Due to the fact that $\vartheta \Psi_2$ is gauge invariant in the Minkowski case, we have 
\begin{align} 
\sigma_p (\vartheta \Psi_2)  \sigma_p (\Kop_0)  = 0.  
\end{align} 
However, it can be shown that, on Kerr, there is no second order operator $\Lop$ with symbol  $\sigma_p(\vartheta\Psi_2)$ such that $\Lop  \Kop_0 = 0$, and hence $\Kop_0$ fails to be involutive.
It was rightly noted in~\cite{pommaret1, pommaret2} that constructing a full compatibility operator for an involutive version of $\Kop_0$ is much easier. However, we point out that our $\Kop_0$ is tied to the fixed notion of gauge symmetry and gauge invariance in linearized gravity, hence we are not free to replace it with its involutive prolongation.   
Preliminary analysis also indicates that neither $\Kop_1$ nor $\widetilde\Kop_1$ is involutive. We do suspect that enlarging
$\tilde{\Kop}_1$ by $\ImA_3$ and $\ImA_4$ defined in Appendix~\ref{sec:AppC2tildeComps} is sufficient to achieve involutivity. However, a full analysis of involutivity goes beyond the scope of this paper.

\section{Discussion} \label{sec:discussion} 
In this paper we have given for the first time a proof of completeness
for a set of gauge invariants for first order metric perturbations of the Kerr spacetime, 
where we have interpreted gauge invariants as compatibility operators for the Killing operator $\Kop$ on
this background. In section~\ref{sec:kerr}, we have constructed an
operator $\Kop_1$ following the methods of~\cite{kh-compat}, which
guarantee that the components of $\Kop_1$ are a complete set of gauge
invariants, even if their explicit expressions turn out to be somewhat
cumbersome. In section~\ref{sec:equiv}, we have shown that the operators
$\Kop_1$ and $\widetilde{\Kop}_1$ factor through each other, where the
$\widetilde{\Kop}_1$ consist of the concise set of gauge invariants introduced
in~\cite{ab-kerr}, thus confirming the completeness of the components of
$\widetilde{\Kop}_1$ that was stated in~\cite{ab-kerr}. With little
extra effort, the construction of section~\ref{sec:kerr} also yielded a
full compatibility complex $\Kop_l$ for $\Kop_0 = \Kop$, terminating
after $l=0,1,2,3$.

There exists a non-linear analog of the problem of constructing a
complete set of linear gauge-invariants on a given background spacetime
$(M,g_{ab})$. Namely, a so-called \emph{IDEAL
characterization}~\cite{coll-ferrando, fs-schw, fs-sphsym, fs-kerr}
of a given spacetime consists of a list of tensors $\{T_i[g]\}$
covariantly built from the metric, Riemann tensor and covariant
derivatives such that the conditions $T_i[g]=0$ are sufficient to
guarantee that $(M,g_{ab})$ is locally isometric to the given reference
spacetime. As was pointed out in the recent works~\cite{cdk,
kh-tangherlini}, where IDEAL characterizations were given for
cosmological FLRW and Schwarzschild-Tangherlini black hole spacetimes,
one can use the tensors $\{T_i[g]\}$ to construct linear gauge
invariants on the characterized spacetime. In particular, the identity
$\Lie_v T_i[g] = \dot{T}^g_i[\Kop^g[v]]$~\cite{stewart-walker}
guarantees that the linear operator $\dot{T}^g_i[h]$ is a gauge
invariant whenever $T_i[g] = 0$ (or even more generally when $T_i[g]$ is
a combination of Kronecker-deltas with constant coefficients).
Conversely, any linear metric perturbation $h_{ab}$ that leaves the
local isometry class must violate the IDEAL characterization equalities,
$\{T_i[g+h+\cdots]\} \ne 0$, which is equivalent to $\{\dot{T}^g_i[h]\}
\ne 0$, unless some $T_i$ vanish to a high order on the space of all
metrics along some directions approaching the reference metric $g$.
Thus, the operators $\{\dot{T}^g_i\}$ have a good geometric
interpretation and provide a good candidate for a complete set of linear
gauge invariants on the reference spacetime geometry. Checking their
equivalence with a systematically constructed complete set of linear
gauge invariants can accomplish a double goal: provide the complete
gauge invariants with a geometric interpretation, and show that the set
$\{\dot{T}^g_i\}$ is indeed complete. Such an exercise has already been
successfully carried out for a class of FLRW geometries~\cite{fhk}. It
would be worth while to complete the comparison, already initiated
in~\cite{ab-kerr}, of the $\widetilde{\Kop}_1$ operator with the linearized
IDEAL characterization of the Kerr spacetime given by Ferrando and
S\'aez~\cite{fs-kerr}.

It is well-known that the construction of Hodge-de~Rham Laplacians on a
Riemannian manifold uses in an essential way the structure of the
de~Rham complex as a compatibility complex. Similarly, it was observed
in~\cite{kh-calabi} that the compatibility complex $\Kop_l$ on a
maximally symmetric Lorentzian spacetime can be endowed with a
Hodge-like structure, producing a sequence of wave-like (normally
hyperbolic~\cite{bgp}) operators $\square_l$, obeying the commutativity
relations $\Kop_l \square_l = \square_{l+1} \Kop_l$. The $\square_l$
operators have several applications: (a) Providing a ``Hodge theory''
for the cohomology of the compatibility complex
$H^*(\Kop_l)$~\cite{kh-causcohom}. (b) Providing a propagation equation
$\square_1 \Kop_1[h]$ directly for the gauge invariants of perturbations
$h$ satisfying the linearized Einstein equations. (c) Providing a
reconstruction of the metric perturbation $h$ from its invariants $\psi
= \Kop_1[h]$, formally $h = \square_0^{-1}\delta_1[\psi] =
\delta_1[\square_1^{-1} \psi]$, where $\square_l = \delta_{l+1}
\Kop_{l+1} + \Kop_l \delta_l$. Alternatively, the metric reconstruction
problem could be locally reduced to an application of the Poincar\'e
lemma to the $\Kop'_l$ complex. It would be interesting to identify such
a ``Hodge-like structure'' also for our $\Kop_l$ compatibility complex
on Kerr.

In section~\ref{sec:counting}, we have discussed the notion of
involutivity and minimality for a differential operator. 
Although it appears that  $\Kop_0, \Kop_1, \widetilde{\Kop}_1$ fail to be involutive, it would be interesting
to find an involutive and minimal version of the full compatibility
complex $\Kop_l$, for $l\ge 1$, for instance by completing
$\widetilde{\Kop}_1$ to involutivity as suggested at the end of
section~\ref{sec:counting} and lifting the rest of the $\Kop_l$
operators in an involutive way. Working with an involutive compatibility
complex can simplify the search for a ``Hodge-like structure'' mentioned
above. In the absence of involution, the differential orders of the
operators $\delta_l$ are not a priori bounded from the known orders of
the $\Kop_l$ and the expected orders of the $\square_l$ operators.

Although the Schwarzschild spacetimes are part of the Kerr family, the
fact that they have a larger number of independent Killing vectors means
that some of the discussion from this paper and the earlier
paper~\cite{ab-kerr} do not apply to them, so they need to be handled as
special cases. In fact, the analogous construction of the compatibility
complex for Schwarzschild spacetimes was already carried out
in~\cite[Sec.3.3]{kh-compat}. Also, in analogy with~\cite{ab-kerr} for
Kerr, a number of convenient sets of linear gauge invariants for
Schwarzschild were given in~\cite{swaab}. It would be interesting to
check whether any of these sets is complete by comparing them to the
complete set of gauge invariants obtained in~\cite[Sec.3.3]{kh-compat}.

Finally, another application of the methods used in this paper would be
the construction of a complete set of linear gauge invariants, as well
as a corresponding full compatibility complex, for the Kerr-Newman
charged rotating black hole spacetime. In the Kerr-Newman case, the
compatibility complex must start with a more complicated operator
$\Kop_0$ that incorporates both the linearized diffeomorphisms and the
electromagnetic gauge transformations.

\subsection*{Acknowledgements}
This work was completed while the authors were in residence at Institut Mittag-Leffler in Djursholm, Sweden during the fall of 2019, supported by the Swedish Research Council under grant no. 2016-06596.
IK was partially supported by the Praemium Academiae of M.~Markl,
GA\v{C}R project 18-07776S and RVO: 67985840. The work of BFW was supported in part by NSF Grants PHY 1314529 and PHY 1607323 to the University of Florida.  Support from the Institut d'Astrphysique de Paris (IAP), where part of this work was carried out, is also acknowledged, as is support at Astroparticule et Cosmologie (APC) from the French state funds managed by the ANR within the Investissements d’Avenir programme under Grants ANR-11-IDEX-0004-02 and ANR-14-CE03-0014-01-E-GRAAL. The authors are grateful to Prof. J.-F. Pommaret for several helpful discussions.
\appendix

\section{Component form of \texorpdfstring{$\widetilde \Cop_2$}{tilde C2}} \label{sec:AppC2tildeComps}
For the frame \eqref{eq:frame} define the connection
\begin{align}
\FrameCon^{i}{}_{jk}\defeq{}&e^{a}{}_{j} e_{b}{}^{i} \nabla_{a}e^{b}{}_{k}.
\end{align}
Define the spinors $\hat{\mathcal{A}}\in \SymSpin_ {1,1}$, $\slashed{P}\in \SymSpin_ {0,0}$, $P\in \SymSpin_ {2,2}$ and $Q\in \SymSpin_ {3,1}$ via
\begin{subequations}
\begin{align}
\hat{\mathcal{A}}\defeq{}&\mathcal{K}^1\sTwist \kappa_{1}{}^4\mathcal{K}^2\mathcal{K}^2\vartheta \Psi
 -  \tfrac{1}{4}\Psi_{2}\kappa_{1}{}^3\xi {\overset{0,0}{\odot}}\slashed{G}
 + \tfrac{2}{3}\kappa_{1}{}^3\xi {\overset{1,0}{\odot}}\mathcal{K}^1\mathcal{K}^2\vartheta \Psi
 -  \tfrac{1}{2}\Psi_{2}\kappa_{1}{}^3\xi {\overset{1,1}{\odot}}\mathcal{K}^0\mathcal{K}^2G,\\
\slashed{P}\defeq{}&\sDiv \hat{\mathcal{A}}
 + \tfrac{1}{2}\kappa_{1}{}^3\Psi_{2}\mathcal{L}_{\xi}\slashed{G},\\
P\defeq{}&\tfrac{1}{2}\Psi_{2}\kappa_{1}{}^3\mathcal{L}_{\xi}G
 + \sTwist \hat{\mathcal{A}},\\
Q\defeq{}&\sCurlDagger \mathcal{P}^{2}\vartheta \Psi
 + \tfrac{5}{6}\mathcal{K}^0\mathcal{K}^2\sCurl \vartheta \Phi
 -  \mathcal{K}^1\mathcal{K}^1\sCurl \vartheta \Phi .
\end{align}
\end{subequations}
Also define tensor versions of $P$, $Q$ and $\kappa$ via
\begin{subequations}
\begin{align}
P_{ab}\defeq{}&\tfrac{1}{4} g_{ab} \slashed{P}
 + P_{ABA'B'} \sigma_{a}{}^{AA'} \sigma_{b}{}^{BB'},\\
Q_{abc}\defeq{}&- \sigma^{AA'}{}_{c} \sigma^{BB'}{}_{a} \sigma^{C}{}_{B'b} Q_{ABCA'} \kappa_{1}{}^3,\\
Y_{ab}\defeq{}&\sigma^{AA'}{}_{a} \sigma^{BB'}{}_{b} \bar\epsilon_{A'B'} \kappa_{AB} \bar{\kappa}_{1'}{},
\end{align}
\end{subequations}
where $\sigma_{a}{}^{AA'}$ is the soldering form.
Due to equations (56) and (58b) in \cite{2016arXiv160106084A} we get the relations
\begin{subequations}
\begin{align}
P={}&\sCurlDagger (\kappa_{1}{}^4\mathcal{K}^1Q)
 + 3 \Psi_{2} \kappa_{1}{}^4\mathcal{K}^1\vartheta \Phi
 -  \tfrac{2}{3}\sTwist (\kappa_{1}{}^4\mathcal{K}^2\sCurl \vartheta \Phi) ,\\
\slashed{P}={}&- \tfrac{4}{3} \kappa_{1}{}^3\mathcal{K}^2\mathcal{K}^2\xi {\overset{0,1}{\odot}}\sCurl \vartheta \Phi.
\end{align}
\end{subequations}
The definition of $\hat{\mathcal{A}}$ only differs by a $\vartheta \Phi$ term compared to $\mathcal{A}$ in \cite{2016arXiv160106084A}. From an argument in that paper one finds that $\ImA$ is gauge invariant.
The components of $\ImA$ can be expressed algebraically in terms of $\InvSymb_\xi$, $\InvSymb_\zeta$ and $P$ via
\begin{subequations}
\begin{align}
\Im \InvSymb_\xi={}&-81 \ImA_{1},\quad
\Im \InvSymb_\zeta={}-81 \ImA_{2},\\
P_{11} -  \bar{P}_{11}={}&\frac{i M \ImA_{4}}{3 \kappa_{1}{} \bar{\kappa}_{1'}{}^3}
 + \frac{i M \ImA_{3}}{3 \kappa_{1}{}^3 \bar{\kappa}_{1'}{}}
 -  \tfrac{2}{81}i \partial_{1}\Im \InvSymb_\xi,\\
P_{12} -  \bar{P}_{12}={}&- \frac{3i M \ImA_{4} (\kappa_{1}{}^3 + \kappa_{1}{} \bar{\kappa}_{1'}{}^2 - 2 \bar{\kappa}_{1'}{}^3)}{4 \kappa_{1}{}^2 \bar{\kappa}_{1'}{}^3}
 -  \frac{3i M \ImA_{3} (-2 \kappa_{1}{}^3 + \kappa_{1}{}^2 \bar{\kappa}_{1'}{} + \bar{\kappa}_{1'}{}^3)}{4 \kappa_{1}{}^3 \bar{\kappa}_{1'}{}^2}\nonumber\\
& -  \tfrac{1}{81}i \partial_{1}\Im \InvSymb_\zeta
 -  \tfrac{1}{81}i \partial_{2}\Im \InvSymb_\xi.
\end{align}
\end{subequations}
Hence, we can conclude that any component of $\ImA$, $P$ or $Q$ can be expressed in terms of $\widetilde\Kop_1$.

For symmetric 2-spinors $\phi, \psi$, set
\begin{align}
\begin{bmatrix} \bar{\phi} \\ \phi \\ \bar{\psi} \\ \psi \end{bmatrix}  = \Kop'_1 \Cop_1
\end{align}
and define the real 2-forms
\begin{subequations}
\begin{align}
G_{ab}\defeq{}&\sigma_{a}{}^{AA'} \sigma_{b}{}^{BB'} (\bar\epsilon_{A'B'} \phi_{AB} + \epsilon_{AB} \bar{\phi}_{A'B'}) \\
H_{ab}\defeq{}&\sigma_{a}{}^{AA'} \sigma_{b}{}^{BB'} (\bar\epsilon_{A'B'} \psi_{AB} + \epsilon_{AB} \bar{\psi}_{A'B'}).
\end{align}
\end{subequations} 
A long computer algebra calculation reveals that this operator factors through $\widetilde{\Kop}_1$ with components \eqref{eq:PRLinv} via
\begin{subequations} 
\begin{align}
M G_{12}={}&
 - 54 (P_{12} + \bar{P}_{12})
 -  \tfrac{2}{3} \partial_{1}\Re \InvSymb_\zeta
 -  \tfrac{2}{3} \partial_{2}\Re \InvSymb_\xi
 +108i \FrameCon^{3}{}_{12} \ImA_{3}
 - 108i \FrameCon^{4}{}_{12} \ImA_{4} ,\\
M G_{13}={}&\tfrac{2}{9} Y_{1}{}^{a} (
\partial_{a}\InvSymb_\xi + 162 P_{1a} 
-2 \FrameCon^{2}{}_{1a} \InvSymb_\zeta - 2 \FrameCon^{1}{}_{1a} \InvSymb_\xi ),\\
M G_{23}={}&
+ \tfrac{2}{9} Y_{1}{}^{a} \bigl(162 P_{2a}
 + \partial_{a}\InvSymb_\zeta
-2i \FrameCon^{2}{}_{2a} \Im \InvSymb_\zeta - 2i \FrameCon^{1}{}_{2a} \Im \InvSymb_\xi  - 54 (Q_{12a} + \bar{Q}_{12a})\bigr)\nonumber\\
&- \tfrac{4}{9} Y_{2}{}^{a}(\FrameCon^{2}{}_{1a} \Re \InvSymb_\zeta + \FrameCon^{1}{}_{1a} \Re \InvSymb_\xi) 
 -  \tfrac{4}{9}i \Im \InvSymb_\xi (\FrameCon^{3}{}_{2a} Y_{3}{}^{a} -  \FrameCon^{4}{}_{2a} Y_{4}{}^{a}),\\
M G_{34}={}&- 36i Y_{4}{}^{a} \partial_{1}\ImA_{a}
 +108i (\FrameCon^{3}{}_{34} - 3 \FrameCon^{4}{}_{33}) \ImA_{3}
 - 108i (3 \FrameCon^{3}{}_{44} -  \FrameCon^{4}{}_{34}) \ImA_{4}\nonumber\\
& - 36 (Q_{341} + \bar{Q}_{341}),\\
M H_{12}={}& - 27 (P_{22} + \bar{P}_{22})
 -  \tfrac{2}{3} \partial_{2}\Re \InvSymb_\zeta
+ 54i \FrameCon^{3}{}_{22} \ImA_{3}
 - 54i \FrameCon^{4}{}_{22} \ImA_{4},\\
M H_{13}={}&
  \tfrac{2}{9} Y_{2}{}^{a} \bigl(162 P_{1a}+ \partial_{a}\InvSymb_\xi
 -i (\FrameCon^{1}{}_{1a} + \FrameCon^{2}{}_{2a}) \Im \InvSymb_\xi  - 27 Q_{1a1} \bigr)
 -18 Q_{231}\nonumber\\
& - 18 (Q_{123} + \bar{Q}_{123})
 -  \tfrac{2}{9} \bigl(2 \FrameCon^{2}{}_{2a} \InvSymb_\zeta + 2 \FrameCon^{1}{}_{2a} \InvSymb_\xi - i (\FrameCon^{1}{}_{1a} + \FrameCon^{2}{}_{2a}) \Im \InvSymb_\zeta\bigr) Y_{1}{}^{a},\\
M H_{23}={}&
 -  \tfrac{2}{9} Y_{2}{}^{a} \bigl( - 162 P_{2a} -  \partial_{a}\InvSymb_\zeta
 +\FrameCon^{2}{}_{2a} (2 \InvSymb_\zeta - i \Im \InvSymb_\zeta) + i \FrameCon^{1}{}_{1a} \Im \InvSymb_\zeta + \FrameCon^{1}{}_{2a} (2 \InvSymb_\xi + i \Im \InvSymb_\xi) + 27 Q_{1a2}\nonumber\\
 &\quad  + 27 (Q_{12a} + \bar{Q}_{12a}) \bigr)
 - \tfrac{4}{3}i \FrameCon^{1}{}_{23} \Im \InvSymb_\zeta
 - 18 Q_{232}
 -  \tfrac{2}{9}i \Im \InvSymb_\xi (\FrameCon^{3}{}_{22} Y_{3}{}^{1} -  \FrameCon^{4}{}_{22} Y_{4}{}^{1}),\\
M H_{34}={}&
  18 \bar{Y}_{2}{}^{a}\bar{P}_{a3} 
 - 18 Y_{2}{}^{a}P_{a4}  
 -  \tfrac{2}{9}i Y_{4}{}^{1} \partial_{1}\Im \InvSymb_\zeta
 + \tfrac{2}{9}i Y_{4}{}^{1} \partial_{2}\Im \InvSymb_\xi
 \nonumber\\
&+ 36i \ImA_{3} (\FrameCon^{1}{}_{2a} \bar{Y}_{1}{}^{a} -  \FrameCon^{1}{}_{1a} \bar{Y}_{2}{}^{a} + \FrameCon^{4}{}_{2a} \bar{Y}_{4}{}^{a})
  + 36i \ImA_{4} (\FrameCon^{1}{}_{2a} Y_{1}{}^{a} -  \FrameCon^{1}{}_{1a} Y_{2}{}^{a} + \FrameCon^{3}{}_{2a} Y_{3}{}^{a})
\nonumber\\
& 
 -9 Q_{234}
 + 27 Q_{243}
 - 27 Q_{342}
 - 9 \bar{Q}_{234}
 - 9 \bar{Q}_{243}
 - 9 \bar{Q}_{342}
 - 81 Q_{144} \kappa_{1}{} \bar{\kappa}_{1'}{}
 + 81 \bar{Q}_{133} \kappa_{1}{} \bar{\kappa}_{1'}{}\nonumber\\
  &+ \frac{81}{8 \kappa_{1}{} \bar{\kappa}_{1'}{}}(\bar{Q}_{144}-  Q_{133}) (\kappa_{1}{}^2 + \bar{\kappa}_{1'}{}^2)^2
    + (-81 Q_{134} + \tfrac{567}{4} Q_{143} - 81 Q_{341} 
  - 81 \bar{Q}_{134}  \nonumber\\
  &\quad + \tfrac{81}{4} \bar{Q}_{143}-  \tfrac{81}{4} \bar{Q}_{341} -  \frac{9 Q_{233}}{2 \kappa_{1}{} \bar{\kappa}_{1'}{}} + \frac{9 \bar{Q}_{244}}{2 \kappa_{1}{} \bar{\kappa}_{1'}{}}) (\kappa_{1}{}^2 + \bar{\kappa}_{1'}{}^2).
\end{align}
\end{subequations}
Here $M = 27 \Psi_{2} \kappa_{1}{}^3$ is the mass parameter of the Kerr solution.


%

\bigskip

\end{document}